\documentclass[12pt,a4paper]{article}

\usepackage{amsfonts}
\usepackage{amssymb}
\usepackage{amsmath}
\usepackage{amsthm}
\usepackage{amscd}
\usepackage{graphicx}
\usepackage{color}
\usepackage[colorlinks=true,citecolor=black,linkcolor=black,urlcolor=blue]{hyperref}
\usepackage{comment}

\theoremstyle{plain}
\newtheorem{theorem}{Theorem}[section]
\newtheorem{lemma}{Lemma}[section]

\theoremstyle{definition}
\newtheorem{definition}{Definition}[section]

\theoremstyle{remark}
\newtheorem{remark}{Remark}[section]

\newcommand{\dd}{\mathrm{d}}
\newcommand{\prt}{\partial}
\newcommand{\be}{\begin{equation}}
\newcommand{\ee}{\end{equation}}
\newcommand{\bea}{\begin{eqnarray}}
\newcommand{\eea}{\end{eqnarray}}
\def\ba#1\ea{\begin{align}#1\end{align}}
\def\bas#1\eas{\begin{align*}#1\end{align*}}
\newcommand{\bd}{\begin{displaymath}}
\newcommand{\ed}{\end{displaymath}}
\newcommand{\nn}{\nonumber\\}
\newcommand{\whr}{\ \textrm{, where} \ }
\newcommand{\mand}{\quad \textrm{and} \quad}

\newcommand{\mt}{\textrm}
\newcommand{\mb}{\mathbb}
\newcommand{\mc}{\mathcal}

\newcommand{\id}{\mathrm{id}}

\newcommand{\st}{\ \textrm{s.t.}\ }
\newcommand{\1}{\mathbf{1}}
\newcommand{\ti}{\tilde}
\newcommand{\res}{\mathrm{Res}}

\definecolor{r}{rgb}{1,0,0}

\title{\bf Combinatorial Hopf algebra for\\the Ben Geloun-Rivasseau tensor field theory}

\author{Matti Raasakka$^{a}$ and Adrian Tanasa$^{a,b}$\\{\tt\small matti.raasakka@lipn.univ-paris13.fr, 
adrian.tanasa@ens-lyon.org}\\
}

\date{}

\begin{document}

\maketitle

\begin{abstract}
The Ben Geloun-Rivasseau quantum field theoretical model is the first tensor model shown to be 
perturbatively renormalizable. 
We define here an appropriate Hopf algebra describing the 
combinatorics of this new tensorial renormalization.
The structure we propose is significantly different from the previously defined Connes-Kreimer 
combinatorial Hopf algebras due to the involved combinatorial and topological properties of the tensorial Feynman graphs. In particular, the 2- and 4-point function insertions must be defined to be 
non-trivial only if the superficial divergence degree of the associated Feynman integral 
is conserved.
\end{abstract}

\newpage

\section{Introduction \& Motivation}
The manuscript at hand concerns the combinatorics of renormalization of tensor field theory (TFT) models. TFT models are a path integral formulation of quantum field theories (QFT), whose Feynman graphs correspond combinatorially to discrete geometries.
In the perturbative expansion of the partition function of a TFT model, to each such Feynman graph is associated a probability (called the Feynman amplitude in QFT), and accordingly we may understand the TFT model as a quantum model of discrete geometry. 
Recently, the first renormalizable rank four TFT model with Lie group valued field variables was formulated and studied by Ben Geloun and Rivasseau \cite{BR}. Since then, the renormalizability of several other TFT models has further been established \cite{BgS,COR,BgL,SVT}.

On the combinatorial level of Feynman graphs the renormalizability property can be phrased in a purely 
algebraic way, as was shown by Connes and Kreimer \cite{CK}. This was followed by the formulation in terms of a Riemann-Hilbert correspondence for categories of
differential systems by Connes and Marcolli \cite{CM, cm2}.

The Connes-Kreimer combinatorial approach to renormalization was initially applied to local QFT models.
When one generalizes to QFT models on non-commutative spaces, such as the Moyal space, locality is lost due to the quantum uncertainty relations between space coordinates.
However, it was shown in \cite{TV,TK} that by replacing the notion of locality with a 
new concept, the so-called `Moyality' property for interactions, the Hopf algebraic structure of Feynman graphs allowing for the renormalization of particular non-commutative models can be established.
The Connes-Kreimer  formalism of matrix field theory, whose two-stranded ribbon Feynman graphs are dual to two-dimensional triangulated surfaces, was also considered in \cite{M}. In addition, Hopf algebraic structures of quantum gravity spin foams, closely related to TFT models, have been considered in \cite{Tan}.
The Hopf algebra of tensorial Feynman graphs of group field 
theory and the associated Schwinger-Dyson equations have also been considered in \cite{Kra}.

The goal of the current paper is then to continue in the above direction of research by formulating the combinatorial Connes-Kreimer Hopf algebraic structure of the Ben Geloun-Rivasseau (BGR) model. Let us emphasize that this constitutes the first application of the Connes-Kreimer combinatorial renormalization to a TFT model of rank-$D$ for $D>2$.

First important difference we note to the usual case of local QFT is the structure of external constraints, namely, the momentum conservation of local QFT is replaced by the so-called tensorial invariance. Moreover, the implementation of the Connes-Kreimer formalism requires particular attention at the level of the $n$-point function insertions. This follows from the fact that the divergencies of the BGR $n$-point functions are not only dependent on the external topological data of the Feynman graph, as in ordinary local QFT models. Thus, only when the superficial divergence degree is conserved, a BGR tensor graph insertion is defined to be non-trivial. This is 
a crucial difference with respect to the previously defined Connes-Kreimer structures for local QFT.

\section{Connes-Kreimer algebra for local QFT}\label{sec:CK}

\subsection{Feynman graphs of quantum field theory}
The basic feature of any QFT model is that it gives a prescription for associating quantum probability amplitudes (taking values in $\mb{C}$) to physical processes. In the path integral formulation of a QFT model, these amplitudes can be extracted from the generating functional $\mc{W}$ of the model that is typically expressed as an infinite dimensional functional integral of the type
\ba
	\mc{W}[J] = \int \mc{D}\phi\ e^{-S[\phi] + \phi\cdot J}\ ,
\ea
where $\phi$ and $J$ denote some fields (usually taken to be smooth sections of a fiber bundle), or collections of fields, on a fixed background manifold $X$. $S[\phi]$ is called the action functional defining the model, and the integral measure can typically be informally written as an infinite product of Lebesgue measures over the field values at each point in $X$, $\mc{D}\phi := \prod_{x\in X} \dd\phi(x)$. Then, the correlation functions $\langle \cdots \rangle$ for the monomials in the field variables, which correspond to the basic physical observables of the model, can be written as functional derivatives of the generating functional,
\ba\label{eq:qftcorrelations}
	\langle \phi(x_1) \cdots \phi(x_n) \rangle \equiv \int \mc{D}\phi\ \phi(x_1) \cdots \phi(x_n) e^{-S[\phi]} = \left. \frac{\delta}{\delta J(x_1)} \cdots \frac{\delta}{\delta J(x_n)} \mc{W}[J] \right|_{J=0}\ .
\ea
The well-definedness of such an expression is, of course, highly questionable. Indeed, apart from rare exceptions, these expressions are highly divergent.

Nevertheless, one can extract finite quantities from these formal divergent expressions, which have been shown to match experimental results with remarkable accuracy. The main tool in this direction is perturbation theory. Typically the action functional is of the form $S[\phi] = S_{\mt{kin}}[\phi] + S_{\mt{int}}[\lambda;\phi]$, where $S_{\mt{kin}}[\phi]$ is quadratic in the field variables, and the higher order polynomial interaction part $S_{\mt{int}}[\lambda;\phi]$ is proportional to a collection of parameters $\lambda$, the coupling constants of the model. Thus, the kinetic part $S_{\mt{kin}}[\phi]$ by itself gives a simple (infinite dimensional) Gaussian measure, whose correlation functions can be computed explicitly. One can then expand the integrand in equation (\ref{eq:qftcorrelations}) in powers of the coupling constant(s) $\lambda$, and in each finite order one ends up with a sum over products of correlation functions of the Gaussian measure with given weights. It is these terms in the perturbative 
expansion, and their combinatorial structure in particular, for which Feynman graphs provide a very convenient bookkeeping device. The contribution to the full correlation function associated to a particular Feynman graph is called the corresponding Feynman amplitude, and can be obtained from the graph by applying the set of Feynman rules of the particular model. Any finite order contribution is then obtained as a sum over the amplitudes of a finite number of Feynman graphs.

However, individual terms in this perturbative series are generically divergent. This particular problem is addressed by the framework of perturbative renormalization. 
The Connes-Kreimer combinatorial Hopf algebra of Feynman graphs is aimed at organizing the renormalization of the individual Feynman amplitudes in the perturbative expansion in a coherent unifying algebraic structure.

\bigskip

Let us then give a standard example of the 
general points discussed above about the relationship of Feynman graphs to QFT models: the $\lambda\phi^4$ scalar field model on $\mb{R}^4$ (with Euclidean metric). This model is determined by the action functional $S[\phi] = S_{\mt{kin}}[\phi] + S_{\mt{int}}[\lambda;\phi]$, where
\ba
\label{act-x}
	S_{\mt{kin}}[\phi] = \frac{1}{2}\int_{\mb{R}^4} d^4x\  \phi (\Delta + m^2) \phi \mand S_{\mt{int}}[\lambda;\phi] = \frac{\lambda}{4!} \int_{\mb{R}^4} d^4x\ \phi^4 \ .
\ea
Here, the single field $\phi$ takes values in $\mb{R}$, and the parameters of the model $m,\lambda\in\mb{R}_+$ correspond to the mass of the particle species modelled and the coupling strength, respectively. 
$\Delta$ denotes the Laplace operator. 
The interaction is local in the sense that it couples the field values only at the same point in space. It is usual to consider the model in the momentum space, obtained through Fourier transform, instead of the direct space, since due to the translation symmetry of the model the total momentum is conserved. 
In momentum space, the action \eqref{act-x} reads explicitly
\ba
	S_{\mt{kin}}[\ti{\phi}] &= \frac{1}{2}\int_{\mb{R}^4} \dd^4p\ \ti{\phi}({p}) (|{p}|^2 + m^2) \ti{\phi}({p}) \,, \\
	S_{\mt{int}}[\lambda;\ti{\phi}] &= \frac{\lambda}{4!} \int_{(\mb{R}^4)^{4}} \big[\prod_{i=1}^4 \dd^4p_i\big] \ti{\phi}({p}_1)\ti{\phi}({p}_2)\ti{\phi}({p}_3)\ti{\phi}({p}_4) \delta^4({p}_1+ {p}_2+ {p}_3+ {p}_4)\,,\label{eq:phi4int}
\ea
where $\ti{\phi}({p}) := \int_{\mb{R}^4}\dd^4x\ e^{-i{p}\cdot{x}} \phi({x})$ is the Fourier transform of the field $\phi$. In the interaction term, the Dirac delta distribution $\delta^4({p}_1+ {p}_2+ {p}_3+ {p}_4)$ imposes the 
so-called {\it momentum conservation}. The covariance of the kinetic part, i.e., the (free) \emph{propagator}, can be explicitly computed to yield
\ba\label{eq:freeprop}
	\mc{P}({p}_1,{p}_2):= \langle \ti{\phi}({p}_1) \ti{\phi}({p}_2) \rangle_{\mt{kin}} = \int \mc{D}\ti{\phi}\ \ti{\phi}({p}_1) \ti{\phi}({p}_2) e^{-S_{\mt{kin}}[\ti{\phi}]} = \frac{1}{|{p}_1|^2+m^2} \delta^4({p}_1-{p}_2) \,.
\ea
Moreover, due to Wick's theorem, any correlation function of such Gaussian measure can be expressed as a sum of products of propagators
\ba\label{eq:Gausscorr}
	\langle \ti{\phi}({p}_1) \ti{\phi}({p}_2) \cdots \ti{\phi}({p}_{2n}) \rangle_{\mt{kin}} = \sum_{\sigma\in\Sigma_{2n}/\sim} \left( \prod_{i=1}^{n} \langle \ti{\phi}({p}_{\sigma(2i-1)}) \ti{\phi}({p}_{\sigma(2i)}) \rangle_{\mt{kin}} \right)\,,
\ea
where $\Sigma_{2n}/\sim$ is the set of $(2n)!$ permutations of $2n$ elements modulo the $n!$ permutations of pairs of elements with indices $(2i-1,2i)$ and the $2^n$ permutations of elements within the pairs, i.e., the set of all possible pairings of indices.\footnote{For example, ${\Sigma_{6}/\sim} = \{(1,2,3,4,5,6),(1,2,3,6,5,4),(1,3,2,4,5,6),(1,3,2,6,5,4),(1,4,3,2,5,6),$ $(1,4,3,6,5,2),(1,5,3,4,2,6),(1,5,3,6,2,4),(1,6,3,4,5,2),(1,6,3,2,5,4)\}$.} For a monomial of odd degree the correlation function vanishes, since the Gaussian measure is symmetric. The QFT perturbation expansion  is based on this  fact. For example, the $k$'th order contribution to the full correlation function of $n$ field values, the \emph{n-point function}, is obtained by expanding the interaction part of the exponential integrand, and picking out the $k$'th order term
\ba
	\langle \ti{\phi}({p}_1) \ti{\phi}({p}_2) \cdots \ti{\phi}({p}_n) \rangle_k = \frac{(-1)^k}{k!} \int \mc{D}\ti{\phi}\ \ti{\phi}({p}_1) \ti{\phi}({p}_2)\cdots \ti{\phi}({p}_2) (S_{\mt{int}}[\lambda;\ti{\phi}])^k e^{-S_{\mt{kin}}[\ti{\phi}]} \,.
\ea
By substituting the explicit expression for the interaction term \eqref{eq:phi4int}, and writing all the resulting correlation functions of the free theory in terms of propagators using \eqref{eq:freeprop} and \eqref{eq:Gausscorr}, one can explicitly calculate this $k$'th order contribution to the $n$-point function of the theory. In QFT, the physical interpretation of this quantity is that it gives the (non-normalized) probability amplitude for a scattering process from an initial state with $m$ free particles having momenta $({p}_i)_{i=1,\ldots,m}$ to a final state with $n-m$ free particles having momenta $({p}_i)_{i=m+1,\ldots,n}$.

\begin{figure}
\centering
\begin{minipage}{0.20\linewidth}\centering
\def\svgwidth{0.8\columnwidth}
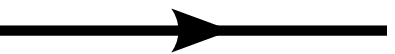
\end{minipage}\hspace{0.05\linewidth}
\begin{minipage}{0.20\linewidth}\centering
\def\svgwidth{0.8\columnwidth}
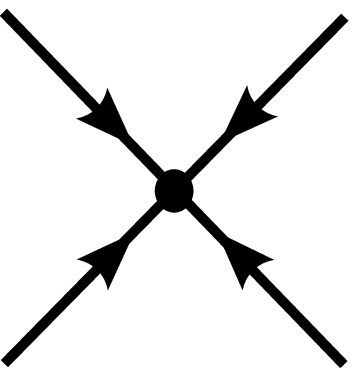
\end{minipage}\hspace{0.1\linewidth}
\begin{minipage}{0.4\linewidth}\centering
\def\svgwidth{0.8\columnwidth}
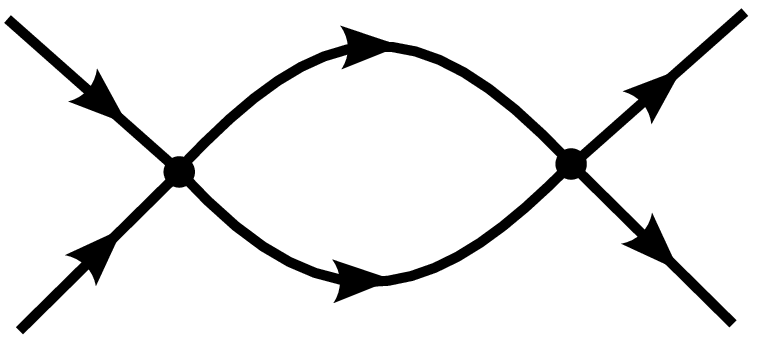
\end{minipage}
\caption{\label{fig:phi4}On the right: The Feynman graph representation of the propagator and the 4-point-interaction for the $\lambda\phi^4$-model. On the left: An example of a second order Feynman graph with four external edges.}
\end{figure}

Let us introduce:
\begin{definition}
A \emph{Feynman graph} $\Gamma$ of a QFT model is a graph that consists of a set of vertices $\Gamma^{[0]}$ and a set of edges $\Gamma^{[1]}$. The edges $\Gamma^{[1]}_{\mt{int}}\subseteq \Gamma^{[1]}$ that are connected to vertices (or to the same vertex) at both ends are called \emph{internal}, while the remaining edges $\Gamma^{[1]}_{\mt{ext}}\subseteq \Gamma^{[1]}$ (either connected to a vertex only at one end, or not connected to a vertex at all) are called \emph{external}. The external edges $e^{ext}_i \in \Gamma^{[1]}_{\mt{ext}}$ are labelled by incoming external momenta $e^{ext}_i \mapsto \vec{p}_i$. 
The external momenta obey the momentum conservation $\sum_i \vec{p}_i=0$. We denote the space of the external momenta of $\Gamma$ by $P_\Gamma \equiv \{ (\vec{p}_i) \in \times_i P_\Gamma^i : \sum_i \vec{p}_i=0\}$, where $\vec{p}_i\in P_i$. 
Let us denote the set of all Feynman graphs of a model by $\mc{G}$
\end{definition}

\bigskip

Let us briefly describe the Feynman graphs and the associated Feynman rules of the $\lambda\phi^4$-model. Each Feynman graph corresponds to an individual summand in the perturbative expansion of the correlation functions with respect to $\lambda$. The Feynman graphs of the $\lambda\phi^4$-model consist of a single type of edges (the propagator) labelled by a momentum variable and four-valent vertices (the interaction), as depicted in Fig.~\ref{fig:phi4}. The number of vertices in a Feynman graph equals the order of the corresponding summand in the perturbative expansion. An arrow on an edge denotes the direction of the momentum flow $p$, and can be assigned arbitrarily; the opposite orientation corresponds to $-p$. In the following we will often just drop the momentum labels and arrows, since their specific assignment is immaterial --- they can always be reintroduced arbitrarily. The Feynman amplitude associated to a particular Feynman graph is constructed as follows:
\begin{enumerate}
	\item Let $p_l$ be the momenta labels on the edges $l\in\Gamma$ of the Feynman graph $\Gamma$. We form the product of free propagators
	\ba
		\prod_{l\in\Gamma} \frac{1}{|p_l|^2 + m^2} \,.
	\ea
	\item For each vertex $v\in\Gamma$, multiply by the coupling constant $\lambda$, and impose momentum conservation by multiplying the above product by
	\ba
		\prod_{v\in\Gamma} \delta^4(\sum_{l\sim v} \sigma_{vl} p_l) \,,
	\ea
	where $l\sim v$ denotes that the edge $l$ is connected to the vertex $v$, and $\sigma_{vl}=\pm 1$ according to whether $p_l$ is incoming or outgoing with respect to the vertex $v$.
	\item Integrate over $p_l$ associated to the internal edges of the graph.
\end{enumerate}
The quantity that results from applying these rules to a Feynman graph is the Feynman amplitude of that graph.

For any Feynman graph, the amplitude factorizes to the product of the amplitudes of its connected parts, so one can just focus only on connected Feynman graphs. On the level of the generating functional, the disconnected graphs can be removed by replacing $\mc{W}[J]$ with the logarithm $\ln \mc{W}[J]$ in \eqref{eq:qftcorrelations}. Moreover, it is common to restrict to consider the so-called one-particle-irreducible (1PI) Feynman graphs, or what are called bridgeless graphs in graph theory language.
\begin{definition}
A Feynman graph $\Gamma$ is called \emph{one-particle-irreducible} (1PI), if $|\Gamma^{[1]}_{\mt{ext}}|>1$ and it is bridgeless, i.e., it cannot be disconnected by removing a single edge. We will also require $|\Gamma^{[1]}_{\mt{int}}|\geq 1$. We denote the set of all 1PI Feynman graphs by $\mc{G}_{1\mt{PI}}$.
\end{definition}
If there did exist a bridge, i.e., an edge connecting two separate parts of the graph, the momentum conservation would fix the momentum flow on this edge, and make it trivial to consider the amplitude in two separate contributions, glued together simply 
through the identification of the momentum variables on both sides of the bridge.

For example, by applying the above Feynman rules to the Feynman graph on the right-hand side of Fig.~\ref{fig:phi4}, we obtain for the corresponding Feynman amplitude the expression
\ba
	& \lambda^2 \int_{\mb{R}^4}\dd^4p_5 \int_{\mb{R}^4}\dd^4p_6 \left[ \prod_{l=1}^{6} \frac{1}{|p_l|^2 + m^2} \right] \delta^4(p_1+p_2-p_5-p_6) \delta^4(p_5+p_6-p_3-p_4) \nn
	=& \left[ \prod_{l=1}^{4} \frac{1}{|p_l|^2 + m^2} \right] \delta^4(p_1+p_2-p_3-p_4) \times \lambda^2 \int_{\mb{R}^4}\dd^4p_5 \frac{1}{|p_5|^2 + m^2} \frac{1}{|p_1+p_2-p_5|^2 + m^2} \,,
\ea
where we integrated over $p_6$. The first part of the bottom line contains the propagators and momentum conservation associated to the external edges of the graph. The second part corresponds to the internal structure, in particular, to the loop formed by edges $5$ and $6$ that allows for one free momentum integration, unconstrained by the momentum conservation at the vertices. In the limit $|p_5|\to\infty$  one has:
\ba
	\lambda^2 \int \dd^4p_5 \frac{1}{|p_5|^2 + m^2} \frac{1}{|p_1+p_2-p_5|^2 + m^2} \sim \lim_{\Lambda\rightarrow\infty} \lambda^2 \int^\Lambda \frac{\dd |p_5|}{|p_5|} \sim  \lambda^2 \lim_{\Lambda\rightarrow\infty}\ln \Lambda \,,
\ea
i.e., the momentum integration related to the loop is logarithmically divergent for large momentum --- thus it exhibits the so-called ultraviolet divergencies of QFT.

This is where renormalization comes in (when possible).
In order to establish renormalizability, of paramount importance is the so-called \emph{power counting theorem} that provides bounds on the divergence of Feynman amplitudes in terms of the combinatorial structure of the corresponding Feynman graphs. 
Denote by $N_{int}$ and $V$ the number of internal edges and the number of vertices, respectively. Then, the number of independent momentum integrations is $N_{int} - V + 1$, and thus the superficial divergence degree in this case reads $\omega := 4(N_{int} - V + 1) - 2N_{int} = 2N_{int} - 
4(V-1)$. On the other hand, we have $4V = 2N_{int} + N_{ext}$, where $N_{ext}$ is the number of external edges, so $\omega = 4 - N_{ext}$. This is the power counting theorem for the $\lambda\phi^4$-model. 
Due to this theorem, only graphs with two or four external edges can be {\it superficially divergent}.

These considerations prompt us to give the following definitions concerning the analytical aspects of Feynman amplitudes associated to Feynman graphs.
\begin{definition}
The \emph{regularized Feynman rules} of a local QFT model is a map $\Gamma \mapsto \phi_\Lambda(\Gamma) \in E_{\Gamma}$ for any $\Gamma\in\mc{G}$, where we denote by $E_{\Gamma}$ the linear vector space of generalized functions on the space $P_\Gamma$ of the external parameters $\vec{p}_i$ associated to the external edges $e^{ext}_i \in \Gamma^{[1]}_{\mt{ext}}$ of $\Gamma$. We call $\phi_\Lambda(\Gamma)$ the \emph{regularized Feynman amplitude} of the graph $\Gamma\in\mc{G}$. The parameter $\Lambda\in\mb{R}_+$ is the \emph{regularization cut-off}, such that in the limit $\Lambda\rightarrow\infty$ we recover the unregularized Feynman rules $\Gamma \mapsto \phi(\Gamma)$.
\end{definition}
\begin{remark}
In the following, we will consider the propagators associated to the external edges of $\Gamma$ not to be included in $\phi(\Gamma)$. Otherwise, one can also simply factor these out of the Feynman amplitudes.
\end{remark}
\begin{definition}
A \emph{renormalization scheme} is specified by a linear operator $R$, which extracts the divergent part of a Feynman amplitude, so that $\lim_{\Lambda\rightarrow\infty} \rho_\Lambda(\Gamma) - R(\rho_\Lambda(\Gamma)) < \infty$.
\end{definition}
\begin{remark}
There is inherent ambiguity in choosing the renormalization scheme, but all the schemes should lead to the same results, when only differences of amplitudes are considered. Also in practise, there are many different choices for implementing the regularization. In order to define the $R$-operator, one may consider the regularized Feynman amplitudes to lie in the space of formal Laurent series in a cut-off parameter $\epsilon \sim \Lambda^{-1}$, and then the $R$-operator projects onto the series with negative powers of $\epsilon$.
Nevertheless, most renormalization scheme used in QFT are not specified by a Rota-Baxter operator
(for example, this is the case of the celebrated BPHZ scheme). The interested reader may refer 
to \cite{patras} for details.
\end{remark}
\begin{definition}
A \emph{superficial divergence degree} is a map $\omega:\mc{G}\rightarrow \mb{Z}$ such that for any Feynman graph $\Gamma$ with $\lim_{\Lambda\rightarrow \infty}\rho_\Lambda(\Gamma)<\infty$ it satisfies $\omega(\gamma)< 0\ \forall \gamma\subset\Gamma$. Let us call $\Gamma\in\mc{G}_{1\mt{PI}}$ with $\omega(\Gamma)\geq 0$ \emph{superficially divergent}, and denote $\mc{G}_{sd}^{\omega} := \{\Gamma\in\mc{G}_{1\mt{PI}} : \omega(\Gamma) \geq 0\}$ the set of superficially divergent 1PI Feynman graphs according to $\omega$.
\end{definition}
\begin{remark}
Such a superficial divergence degree is usually obtained via a \emph{power counting theorem}, which relates features of the combinatorial and topological structure of Feynman graphs to the divergence of the corresponding Feynman amplitudes. Notice that the superficial divergence degree only sets bounds on the divergence. That is, not all superficially divergent graphs necessarily have divergent Feynman amplitudes. Conversely, a superficially convergent graph may have divergent subgraphs.
\end{remark}

\subsection{Hopf algebra of Feynman graphs}
\label{subsec:HopfFeynman}

Let us then proceed to describe the Connes-Kreimer algebra for local QFT. A Hopf algebra over $\mb{C}$ is a tuple $(\mc{H},u,m,\epsilon,\Delta,S)$, where
\begin{itemize}
	\item $\mc{H}$ denotes the set of elements,
	\item $u:\mb{C}\rightarrow \mc{H}$ denotes the unit,
	\item $m: \mc{H} \otimes \mc{H} \rightarrow \mc{H}$ denotes the product,
	\item $\epsilon: \mc{H} \mapsto \mb{C}$ denotes the counit,
	\item $\Delta: \mc{H} \rightarrow \mc{H} \otimes \mc{H}$ denotes the coproduct, and
	\item $S:\mc{H}\rightarrow \mc{H}$ denotes the antipode.
\end{itemize}
These objects are required to satisfy several compatibility conditions in order for $(\mc{H},u,m,\epsilon,\Delta,S)$ to form a Hopf algebra. (See, e.g., \cite{Manchon} for an exposition of Hopf algebras.)


Let us then paraphrase a fundamental result:
\begin{lemma}[\cite{Manchon}, Corollary II.3.2.]
A connected graded bialgebra is a Hopf algebra. The antipode map $S:\mc{H}\rightarrow\mc{H}$ may be obtained through the recursive formula
\ba\label{eq:Srecurs}
	S(\Gamma) = -\Gamma - m\circ (\mathrm{id}\otimes S)\circ\Delta'(\Gamma)
\ea
or, alternatively,
\ba\label{eq:Srecurs2}
	S(\Gamma) = -\Gamma - m\circ (S \otimes \mathrm{id})\circ\Delta'(\Gamma)
\ea
by setting $S(\1)=\1$.
\end{lemma}

\bigskip

We now proceed to define the corresponding operations on a set of Feynman graphs of a given QFT model, here the $\lambda\phi^4$ scalar field model already introduced in the previous subsection.
\begin{definition}\label{def:res}
The \emph{residue} $\res(\Gamma)$ of $\Gamma\in\mc{G}$ is the graph obtained by contracting all internal edges of $\Gamma$ to a point (or points, one for each connected component, if $\Gamma$ is disconnected). Clearly, $P_{\res(\Gamma)}=P_\Gamma$, since any Feynman graph of a local QFT model satisfies momentum conservation.
\end{definition}
\begin{definition}
	A graph $\gamma\in\mc{G}$ is a (proper) \emph{subgraph} of a graph $\Gamma\in\mc{G}$, denoted by $\gamma\subseteq\Gamma$ ($\gamma\subsetneq\Gamma$), if its internal edges are a (proper) subset of the internal edges of $\Gamma$. Two subgraphs are \emph{disjoint}, if they do not share any internal edges.
\end{definition}

Now, let us introduce some algebraic structure on the set $\mc{H}$ of disjoint unions of 1PI Feynman graphs of a local QFT model. Consider the associative commutative product $m:\mc{H}\otimes\mc{H}\rightarrow\mc{H}$ of graphs as given by the disjoint union $m(\Gamma\otimes\Gamma')=\Gamma \cup \Gamma'$. The unit with respect to this product is given by $u: \mb{C} \rightarrow \mc{H}$ such that $u(1)=\1$, where $\1\in\mc{H}$ denotes the empty graph. Thus, $\mc{H} \equiv \mb{C}[\1,\mc{G}_{1\mt{PI}}]$, the $\mb{C}$-module of polynomials generated by the elements of $\mc{G}_{1\mt{PI}}$ and the empty graph $\1$, equipped with the above operations linearly extended to $\mc{H}$ constitutes a unital associative algebra.

Next, let us formulate a coalgebra structure on $\mc{H}$. First we need to define the operation of subgraph contraction.
\begin{definition}\label{def:subcontr}
Define the operation of \emph{subgraph contraction} as follows: For $\gamma \subset \Gamma$ such that $\res(\gamma)$ corresponds to an interaction vertex of the model, let $\Gamma/\gamma\in\mc{G}$ be the graph that is obtained from $\Gamma$ by replacing $\gamma\subset\Gamma$ with $\res(\gamma)$ inside $\Gamma$. This definition of subgraph contraction has an obvious extension to the case of $\gamma$ being a disjoint union $\gamma=\cup_{i=1}^n\gamma_i$ of disjoint subgraphs $\gamma_i\subset\Gamma$ ($|(\gamma_i)^{[1]}_{int} \cap (\gamma_j)^{[1]}_{int}| = 0$), such that $\Gamma/\gamma \equiv (\cdots ((\Gamma/\gamma_1)/\gamma_2)\cdots /\gamma_n)$ is obtained from $\Gamma$ by replacing all $\gamma_i$ by $\res(\gamma_i)$ inside $\Gamma$.
\end{definition}
\begin{figure}
\centering
\def\svgwidth{0.9\columnwidth}
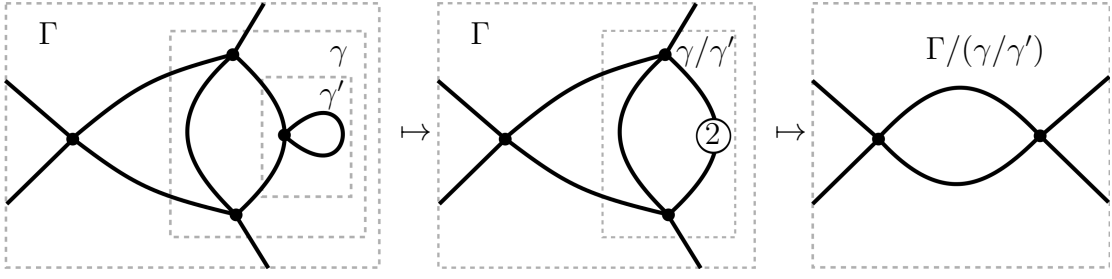
\caption{\label{fig:contrphi4}An example of subgraph contractions for $\lambda\phi^4$-model.}
\end{figure}

Now, we may introduce an operation $\Delta$ acting on the generators $\Gamma\in\mc{G}$ of $\mc{H}$ as
\ba
	\Delta(\Gamma) = \Gamma \otimes \1 + \1 \otimes \Gamma + \sum_{\substack{\gamma\in\cup\mc{G}_{sd}^\omega\\\gamma\subsetneq\Gamma}} \gamma \otimes \Gamma/\gamma \,.
\ea
Notice that here the sum runs over all disjoint unions $\gamma=\cup_i\gamma_i\in \cup\mc{G}_{sd}^\omega$ of superficially divergent disjoint 1PI subgraphs $\gamma_i$ of $\Gamma$. Since $\Delta$ is an algebra homomorphism, it may be extended to $\mc{H}$. It may be proved
\begin{lemma}[\cite{CM}, Theorem 1.27]
	$\Delta: \mc{H} \rightarrow \mc{H}\otimes\mc{H}$ is a coassociative coproduct.
\end{lemma}
The counit of $\mc{H}$ with respect to $\Delta$ is given by the linear extension of $\epsilon: \mc{G} \rightarrow \mb{C}$ such that $\epsilon(\1)=1$ and $\epsilon(\Gamma)=0$ for all $\Gamma\neq\1$. $\mc{H}$ equipped with the above structure thus forms a bialgebra.
\begin{figure}
\centering
\def\svgwidth{\columnwidth}
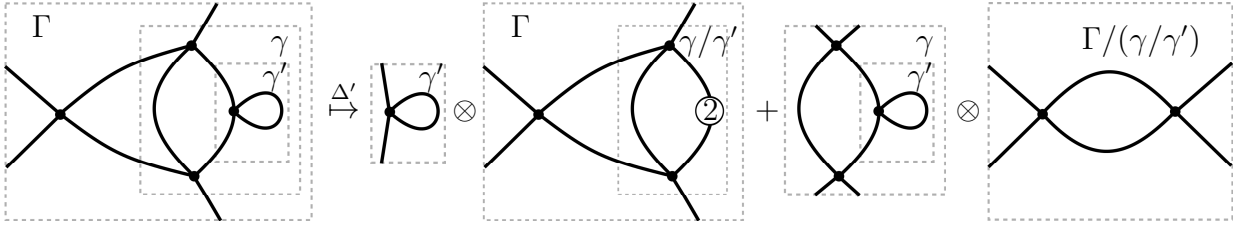
\caption{\label{fig:coprodphi4}An example of the coproduct structure for $\lambda\phi^4$-model.}
\end{figure}

Moreover, we have a natural grading for the elements of $\mc{H}$ given by the number of internal edges, which is compatible with the coproduct $\Delta$. Clearly, $\mc{H}^{(0)} = \mb{C}\1 = u(\mb{C})$. Thus, $\mc{H}$ constitutes a graded connected bialgebra, and accordingly we have
\begin{theorem}[\cite{CM}, Theorem 1.27]
	$(\mc{H},u,m,\epsilon,\Delta,S)$ is a Hopf algebra, where the antipode $S$ is given by the formula (\ref{eq:Srecurs}) (or (\ref{eq:Srecurs2})).
\end{theorem}

Consequently, the combinatorial algebraic properties of Feynman graphs and amplitudes join their forces in the following theorem by Connes and Kreimer \cite{CK}.
\begin{theorem}[\cite{CK}]\label{thm:ren}
	For a local perturbatively renormalizable QFT model the renormalized Feynman amplitudes are given by the formula
\ba
	\phi_R(\Gamma) = S^\phi_R(\phi(\Gamma)) \star \phi(\Gamma) \,,
\ea
where
\ba
	S^\phi_R(\phi(\Gamma)) = - R[\phi(\Gamma)] - R\left[\sum_{\substack{\gamma\subsetneq\Gamma\\\gamma\in\mc{G}_{sd}^\omega}} S^\phi_R(\phi(\gamma)) \phi(\Gamma/\gamma) \right]
\ea
is given through recursion. Here, the $R$-operator defines the corresponding renormalization scheme. Notice that we have $S^\phi_R(\phi(\Gamma)) = -R(\phi(\Gamma))$ for $\Gamma$ without superficially divergent subgraphs, which facilitates the recursion.
\end{theorem}

Finally, let us describe the dual operation to subgraph contraction.
\begin{definition}\label{def:insert}
The operation of \emph{subgraph insertion} is defined as follows: Let $\Gamma\in\mc{G}$ and $v\in\Gamma^{[0]}$. Then, for $\gamma\in\mc{G}^\omega_{sd}$ such that $\res(\gamma)\sim v$, denoting $\res(\gamma)$ is of the type of $v$, we define $\gamma \circ_{v} \Gamma$ as the graph obtained from $\Gamma$ by replacing $v$ with $\gamma$. Let us also define
\ba\label{eq:insum}
	\gamma \circ \Gamma = \sum_{\substack{v\in\Gamma^{[0]}\\v\sim\res(\gamma)}} \gamma \circ_{v} \Gamma \,.
\ea
Notice that, in general, there may be several inequivalent ways of inserting $\gamma$ into a vertex $v$ of $\Gamma$. In this case we must provide additional gluing data and sum also over all different ways of insertion in (\ref{eq:insum}).
\end{definition}

In the $\lambda\phi^4$-model considered above Feynman graphs consist originally of a single type of edge and a single type of 4-valent vertex. The external edges are labelled by the external 4-momentum variables $\vec{p}_i$, which for any connected component of a Feynman graph satisfy the momentum conservation constraint $\sum_i \vec{p}_i = 0$. It is important in the view of subgraph insertions and contractions to note the almost trivial fact that any graph consisting of vertices, which satisfy the momentum conservation constraint, and propagators, which preserve momentum, satisfies itself the momentum conservation constraint.

As we have already seen, the superficial degree of divergence for a Feynman graph $\Gamma$ of the $\lambda\phi^4$-model is given by the formula $\omega=4-N_{ext}$, where $N_{ext}:=|\Gamma^{[1]}_{ext}|$ denotes the cardinality of $\Gamma^{[1]}_{ext}$. Thus, only graphs with two or four external edges need to be renormalized.
Accordingly, we must consider in addition two types of $2$-valent vertices corresponding to mass and wave function renormalization counter-terms, which arise as residues of superficially divergent graphs. 
We decorate them with the label "2",  as an indication of the quadratic character of the corresponding divergence.
The notion of external structure in local QFT is introduced in order to distinguish the contributions of 2-point functions to mass and wave-function renormalization.
\begin{definition}
The \emph{external structure} of a Feynman graph $\Gamma$ is specified by the action of a set of distributions $\{\sigma_v \in E_\Gamma^*\}$ labelled by the vertices of the model, such that $\langle \sigma_v,\phi(\Gamma)\rangle = \rho_v(\Gamma) \langle \sigma_v,\phi(\res(\Gamma))\rangle$, where $\rho_v(\Gamma) \in E_\Gamma$ are characters of $\mc{H}$, and $\langle\sigma_v,\phi(\res(\Gamma))\rangle$ is the kernel of the interaction functional associated to $v$.
\end{definition}
However, we do not make a distinction here at the level of graph drawing. For a graph with two external edges, its contributions to wave function and mass renormalization are given by the external structures for each corresponding vertex. In Figures \ref{fig:contrphi4} and \ref{fig:coprodphi4} we have provided some basic examples of the subgraph contraction operation and the coproduct, 
respectively, for $\lambda\phi^4$-model.

\bigskip

Before ending this section, let us emphasize the following fact of particular 
importance for the sequel.
The 2-valent vertices are associated with a quadratic divergence, since the power-counting simply gives $\omega=4-N_{ext}=2$ for them. Due to this simple fact one may insert any graph with two external edges into a bare propagator of any Feynman graph without changing the divergence degree of the graph, because one simultaneously needs to add a propagator, which adds a counter-balancing $-2$ to the divergence degree. Similarly, all 4-valent graphs are associated with a logarithmic divergence, which allows one to insert any graph with four external edges into a vertex without affecting the divergence degree. We already anticipate that this simple behavior does not hold for models with more complicated power-counting, such as the BGR model, in which case the notions of subgraph contraction insertion must be treated more carefully.

\section{The Ben Geloun-Rivasseau tensor field theory model}\label{sec:BGR}

The form of the BGR model derives from the combinatorics of 4-dimensional simplicial geometry. Consider a 4-simplex. Its boundary is given by five 3-simplices, i.e., tetrahedra. Each pair of these boundary tetrahedra shares one boundary triangle. This combinatorial structure can be illustrated by a stranded graph, as in Fig.~\ref{fig:colorvertex}. Here, each individual strand represents a triangle, while the five rectangular boxes, each connected to four strands entering the graph, represent the five boundary tetrahedra. The tetrahedra are considered colored, and we allow for two opposite orientations of the coloring. 
One considers colored tetrahedra and only two orientations of the coloring in order to simplify the combinatorics and the topology of this type of models (see, for example,  the paper \cite{Gurau}, where this type of model was first introduced).

One can then build larger stranded graphs from these building blocks by identifying boundary tetrahedra with matching colors. Moreover, we allow only for identification of tetrahedra belonging to 4-simplices of opposite orientations, taken into account by alignment of arrows in the boxes representing tetrahedra. 
The resulting stranded graph will then be combinatorially dual to a 4-dimensional simplicial pseudo-manifold. 
\begin{figure}
\centering
\def\svgwidth{0.6\columnwidth}
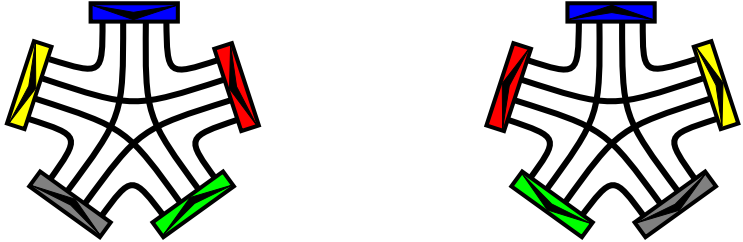
\caption{\label{fig:colorvertex}The two orientations of the stranded representation of the combinatorics of a 4-simplex, where we consider the colors to be: blue, red, green, grey and yellow.}
\end{figure}

A (quantum) statistical description of such simplicial geometries can be formulated as a colored tensor model.
 This is a QFT model, whose Feynman graphs have exactly such stranded structure. 
The model has fields of five different colors (each colored block in Fig. \ref{fig:colorvertex}
 being a field)
and each such field is  a complex-valued function on $U(1)^{\times 4}$.
The two stranded graphs in Fig.~\ref{fig:colorvertex} correspond to the two interaction terms of the model, while the propagator simply identifies the strands of the same color associated to interaction terms with opposite orientations. 
Such models are sometimes called {\it combinatorially non-local} due to the combinatorial nature of the non-locality of their interaction terms.

The Feynman graphs of the BGR model are obtained from such stranded representations of simplicial geometry by removing four of the five fields (namely the blue, the red, the green and the yellow fields in Fig. \ref{fig:colorvertex} above). Thus, one considers interaction terms
such as those illustrated in Fig.~\ref{fig:BGRvertices6} and Fig.~\ref{fig:BGRvertex41}. 
These terms correspond to one-to-one relations between particular colored graphs and 
new ``uncolored'' vertices which, by construction, only have grey field on the boundaries. 

These vertices are dual to four-dimensional polytopes, whose boundaries are composed of six or four tetrahedra. Again, the combinatorial connectivity of the strands represents the identifications of boundary triangles of the boundary tetrahedra.

Note that these are not the only possible ``uncolored'' vertices that one can construct. However,
these are the only type of vertices considered in the BGR model in order to ensure renormalizability (see again \cite{BR} for details).

In addition, a non-trivial kinetic term is added in order to induce dynamics, which allow for renormalization (see Fig. \ref{propa} for a graphical representation). Note that propagators cannot shuffle strands.

\begin{figure}
\centering
\def\svgwidth{0.25\columnwidth}
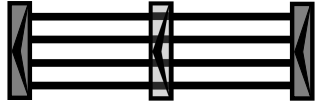
\caption{\label{propa}The kinetic term.}
\end{figure}

In Fig.~\ref{fig:BGRgraph1} we illustrate a simple example of a Feynman graph of the BGR model, where the propagators are represented by transparent rectangles, while the external lines end in solid rectangles.
\begin{figure}
\centering
\def\svgwidth{0.6\columnwidth}
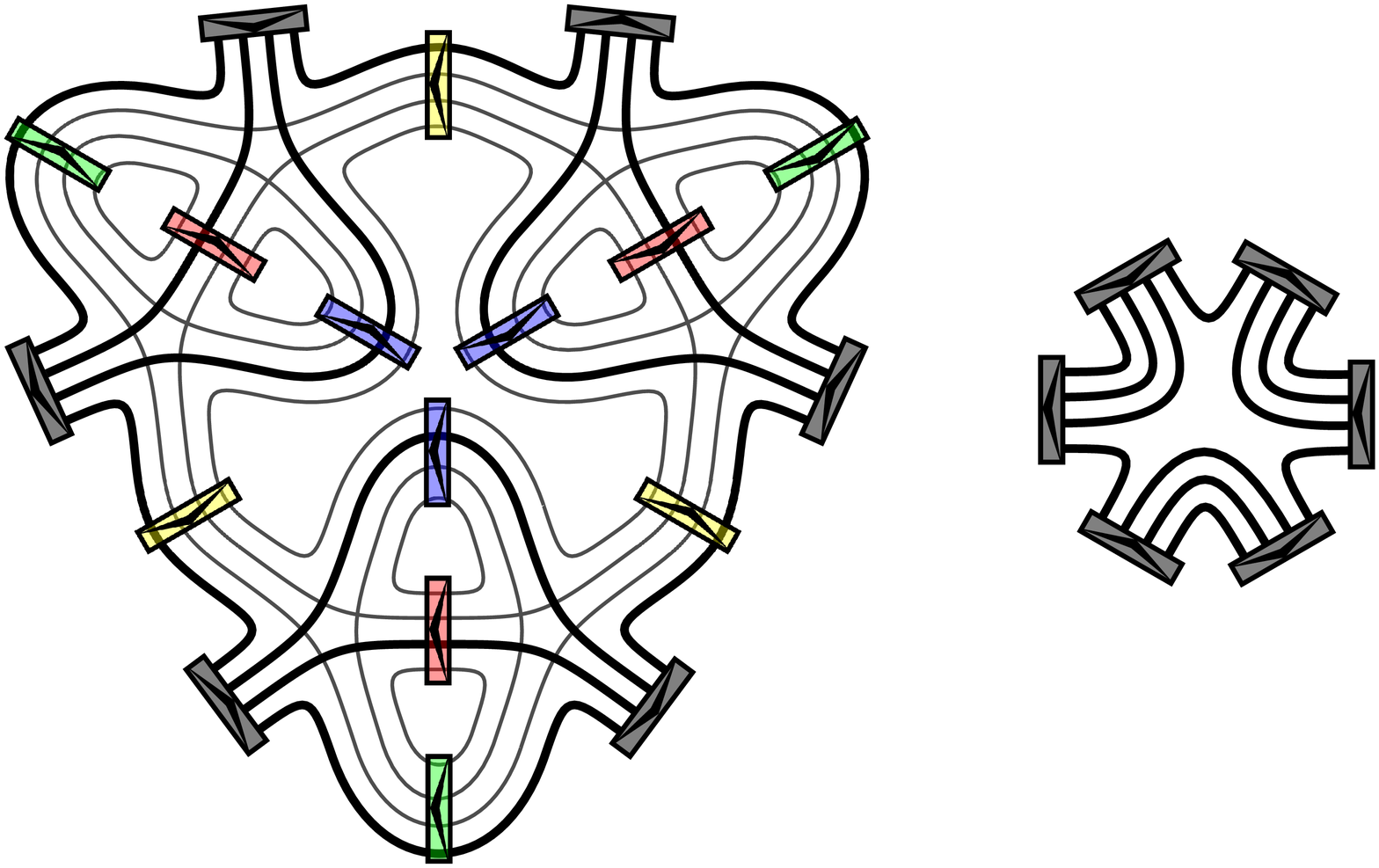
\vspace{10pt}
\def\svgwidth{0.8\columnwidth}
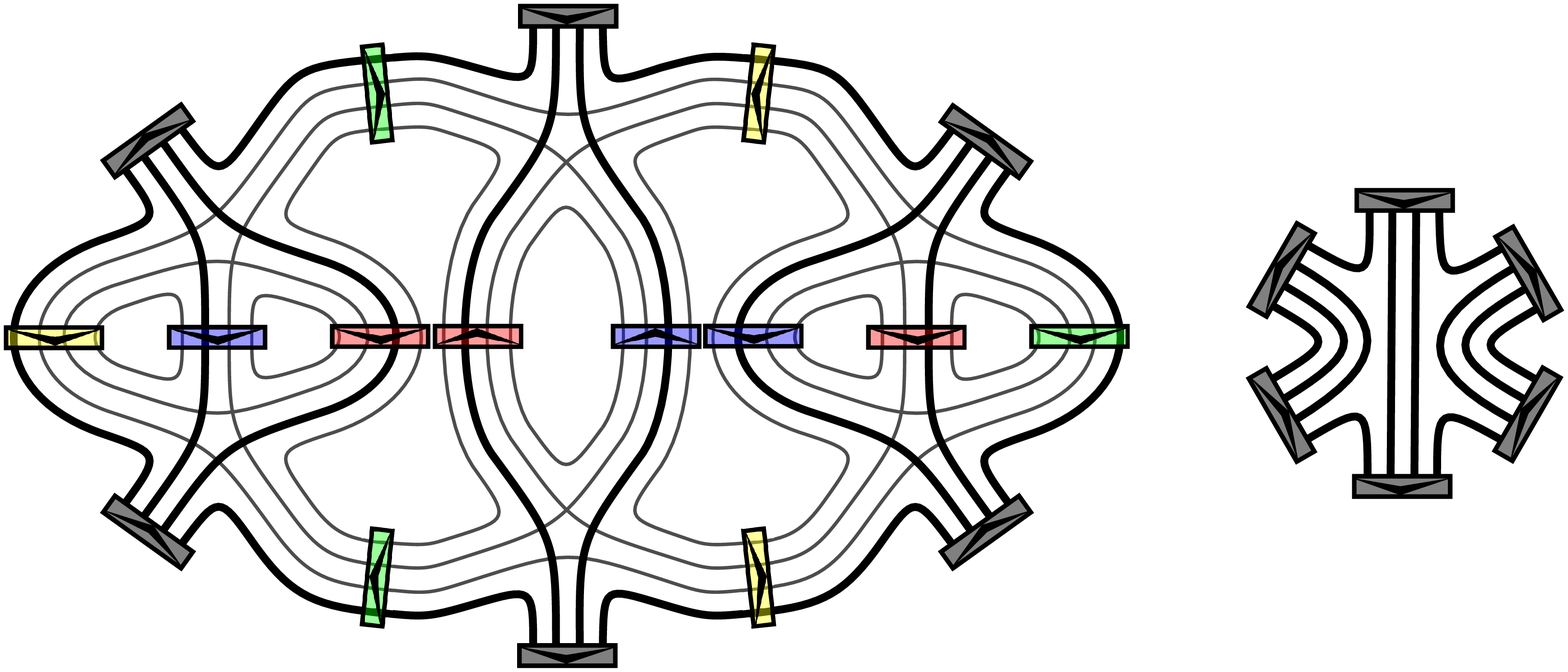
\caption{\label{fig:BGRvertices6}Stranded structures of the vertices $V_{6;1}$ (upper) and $V_{6;2}$ (lower) of the BGR model.} 
\end{figure}
\begin{figure}
\centering
\def\svgwidth{0.7\columnwidth}
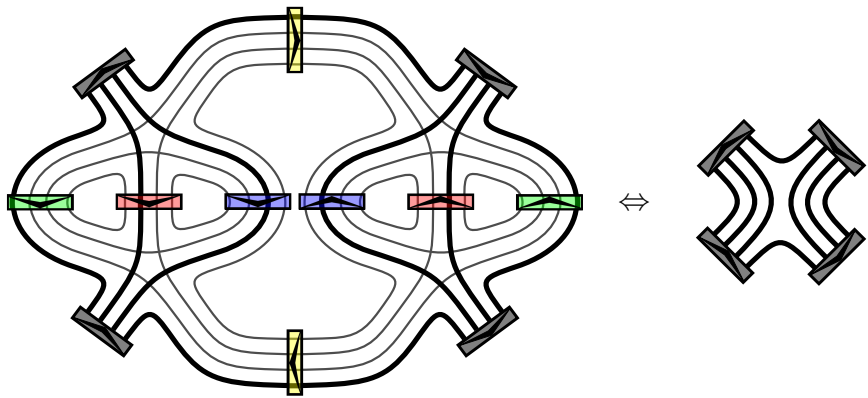
\caption{\label{fig:BGRvertex41}Stranded structure of the vertex $V_{4;1}$ of the BGR model.}
\end{figure}
\begin{figure}
\centering
\def\svgwidth{0.4\columnwidth}
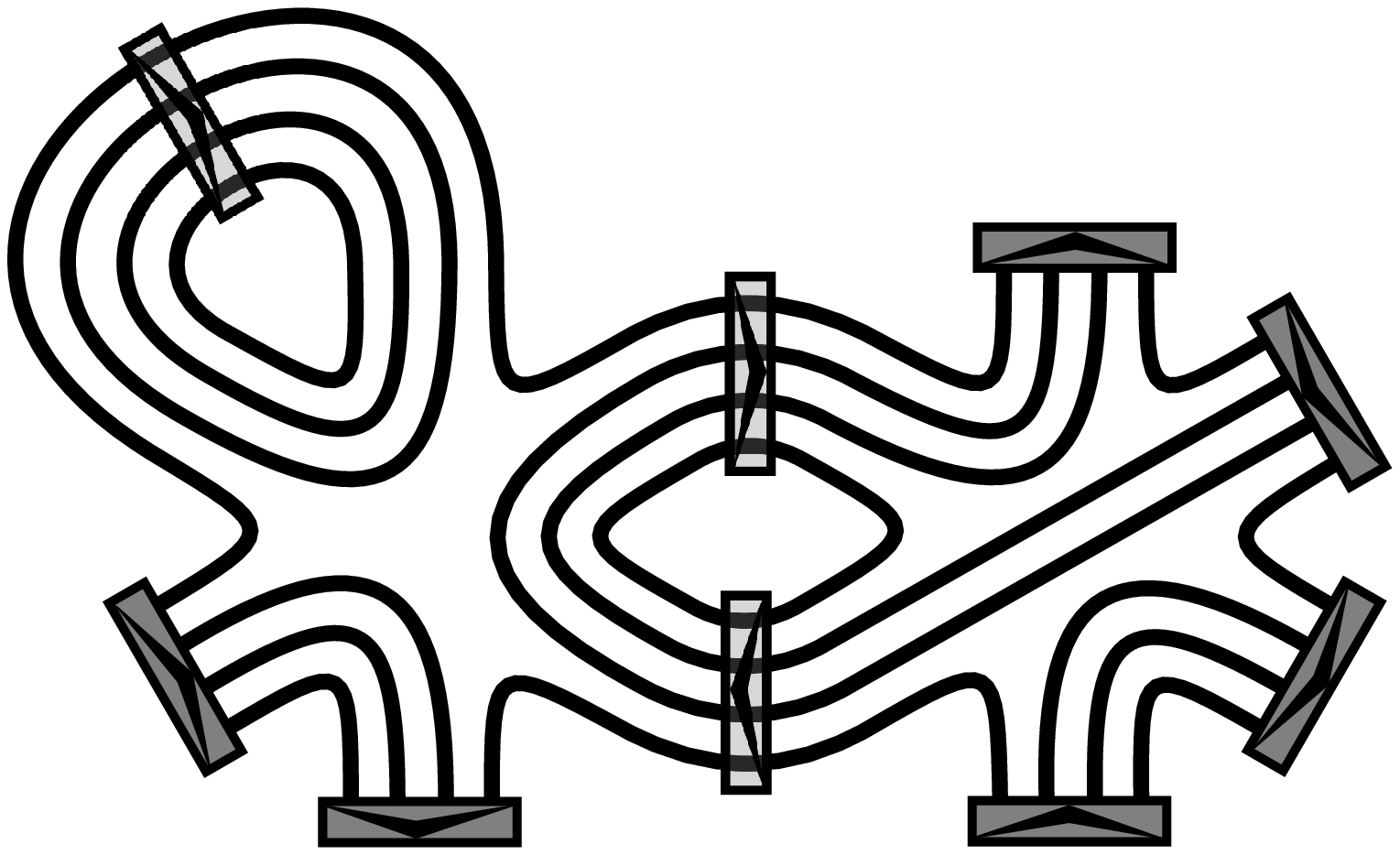
\caption{\label{fig:BGRgraph1}A simple example of a Feynman graph of the BGR model. There is a 
$V_{6;1}$ vertex on the left and a $V_{6;2}$ vertex on the right. $\vec{k},\vec{l},\vec{m},\vec{n},\vec{p},\vec{q}\in\mb{Z}^4$ label the external momenta, while $\vec{r},\vec{s},\vec{t}\in\mb{Z}^4$ label momenta running through internal propagators. The momentum components $k_i\in\mb{Z}$, $i=1,2,3,4$, are associated to the strands as indicated by the numbers (and similarly for the other momenta).}
\end{figure}

Let us then describe in detail the action functional of the BGR model. The fields of the model live on a direct product of four copies of the group $U(1)$. The momentum representation considered in the following is obtained through harmonic analysis, which yields the momentum space $\mb{Z}^4$. In the momentum representation the kinetic term of the action functional reads explicitly
\ba
	S_{kin}[\phi] = \sum_{\vec{p} \in \mb{Z}^4} \overline{\phi(\vec{p})}(|\vec{p}|^2 + m^2)\phi(\vec{p}) \,,
\ea
where $|p|^2 = \sum_{i=1}^4 |p_i|^2$. This leads to the propagator
\ba
	\mc{P}(\vec{p},\vec{q}) = \frac{1}{\sum_i |p_i|^2 + m^2} \prod_{i=1}^4 \delta_{p_iq_i}.
\ea
In the direct ({\it i.e.} $U(1)$) representation, by Fourier transform, the kinetic term reads
\ba
S_{kin}[\tilde \phi]=\int \overline{\tilde\phi}_{1,2,3,4}\left( - \sum_{s=1}^4 \Delta_s +m^2\right)
\tilde \phi_{1,2,3,4},
\ea
where the integral is performed for the Haar measure over the respective group variables, 
$\Delta_s$ denotes the Laplacian operator on $U(1)$ acting on the strand index $s$.
Note that $\tilde \phi$ denotes the Fourier transform of the function $\phi$.

The interaction terms are based on a completely different principle --- instead of locality, they are build on the notion of simplicial combinatorics, as described above. The interaction terms read
\ba
	S_{int}^{6;1}[\phi] = \lambda_{6;1}\sum_{\vec{p},\vec{p}',\vec{p}'' \in \mb{Z}^4} \sum_{\sigma\in\Sigma_4} & \phi_{\sigma_1,\sigma_2,\sigma_3,\sigma_4}\, \overline{\phi_{\sigma_1',\sigma_2,\sigma_3,\sigma_4}}\, \phi_{\sigma_1',\sigma_2',\sigma_3',\sigma_4'}\, \overline{\phi_{\sigma_1'',\sigma_2',\sigma_3',\sigma_4'}}\,\nn
	 & \times \phi_{\sigma_1'',\sigma_2'',\sigma_3'',\sigma_4''}\, \overline{\phi_{\sigma_1,\sigma_2'',\sigma_3'',\sigma_4''}} \,,
\ea
and
\ba
	S_{int}^{6;2}[\phi] = \lambda_{6;2}\sum_{\vec{p},\vec{p}',\vec{p}'' \in \mb{Z}^4} \sum_{\sigma\in\Sigma_4} & \phi_{\sigma_1,\sigma_2,\sigma_3,\sigma_4}\, \overline{\phi_{\sigma_1',\sigma_2',\sigma_3',\sigma_4}}\, \phi_{\sigma_1',\sigma_2',\sigma_3',\sigma_4'}\, \overline{\phi_{\sigma_1'',\sigma_2,\sigma_3,\sigma_4'}}\, \nn
	& \times \phi_{\sigma_1'',\sigma_2'',\sigma_3'',\sigma_4''}\, \overline{\phi_{\sigma_1,\sigma_2'',\sigma_3'',\sigma_4''}} \,,
\ea
where we denoted for simplicity $\phi(p_1,p_2,p_3,p_4)=:\phi_{1,2,3,4}$, $\phi(p_1',p_2,p_3,p_4)=:\phi_{1',2,3,4}$, $\phi(p_1,p_2',p_3,p_4)=:\phi_{2',1,3,4}$, etc., and we sum over all permutations $\sigma\in\Sigma_4$ of the four color indices. These interactions lead exactly to the identifications of the field variables represented by the stranded Feynman graphs of the type in Fig.~\ref{fig:BGRvertices6} modulo permutations of the ordering of the strands.

In addition to the above 6-valent interaction vertices, there are also two 4-valent vertices in this model. The first one is depicted in Fig.~\ref{fig:BGRvertex41}, while the second one is a disconnected combination of two propagators, which arises from the vertex $V_{6;2}$ through the insertion of a propagator, as in Fig.~\ref{fig:BGRvertex42}.
The corresponding interaction terms of the action functional read
\ba
	S_{int}^{4;1}[\phi] = \lambda_{4;1} \sum_{\vec{p},\vec{p}' \in \mb{Z}^4} \sum_{\sigma\in\Sigma_4} \phi_{\sigma_1,\sigma_2,\sigma_3,\sigma_4}\, \overline{\phi_{\sigma_1',\sigma_2,\sigma_3,\sigma_4}}\, \phi_{\sigma_1',\sigma_2',\sigma_3',\sigma_4'}\, \overline{\phi_{\sigma_1,\sigma_2',\sigma_3',\sigma_4'}} \,,
\ea
and
\ba
	S_{int}^{4;2}[\phi] = \lambda_{4;2} \sum_{\vec{p},\vec{p}' \in \mb{Z}^4} \sum_{\sigma\in\Sigma_4} \phi_{\sigma_1,\sigma_2,\sigma_3,\sigma_4}\, \overline{\phi_{\sigma_1,\sigma_2,\sigma_3,\sigma_4}}\, \phi_{\sigma_1',\sigma_2',\sigma_3',\sigma_4'}\, \overline{\phi_{\sigma_1',\sigma_2',\sigma_3',\sigma_4'}} \,.
\ea
\begin{figure}
\centering
\def\svgwidth{0.45\columnwidth}
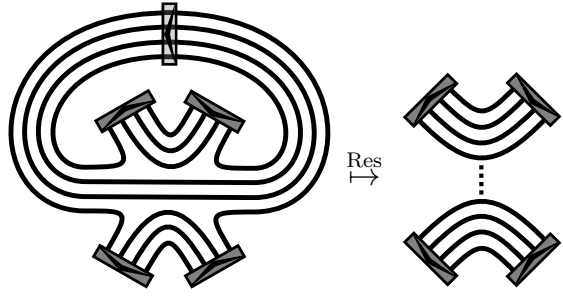
\caption{\label{fig:BGRvertex42}Vertex $V_{4;2}$ of the BGR model obtained through the contraction of a self-loop from the vertex $V_{6;2}$.}
\end{figure}
Notice the peculiar disconnected structure of the vertex $V_{4;2}$. This 
can be interpreted geometrically as a four-sphere with two holes of the form of four-simplices.

The Feynman amplitude associated to a particular Feynman graph is constructed as follows:
\begin{enumerate}
	\item Let $\vec{p}_e\in\mb{Z}^4$ be the momenta labels on the lines $e\in\Gamma^{[1]}$ of the Feynman graph $\Gamma$.  We 
	construct the product of free propagators
	\ba
		\prod_{e\in\Gamma^{[1]}} \frac{1}{|\vec{p}_e|^2 + m^2} \,.
	\ea
	\item For each vertex $v_t\in\Gamma^{[0]}$, multiply by the coupling constant $\lambda_t$, where $t$ labels the type of the vertex, and impose identifications of strands (momentum variables) by multiplying the above product by
	\ba
		\prod_{v\in\Gamma^{[0]}} \prod_k \delta(c^{t(v)}_k(\vec{p}_l)) \,,
	\ea
	where $\{c_{t(v)}^k(\vec{p}_l)=0\}$ is the set of constraints imposed on the incoming momenta by the vertex type $t(v)$ of vertex $v\in\Gamma^{[0]}$. Explicitly, these are given by the following
	\begin{itemize}
		\item $\prod_k \delta(c^{6;1}_k(\vec{p}_l)) = \delta_{p_1n_1}\delta_{p_2q_2}\delta_{p_3q_3}\delta_{p_4q_4}\delta_{q_1k_1}\delta_{k_2l_2}\delta_{k_3l_3}\delta_{k_4l_4}\delta_{l_1m_1}\delta_{m_2n_2}\delta_{m_3n_3}\delta_{m_4n_4}$,
		\item $\prod_k \delta(c^{6;2}_k(\vec{p}_l)) = \delta_{p_1n_1}\delta_{p_2l_2}\delta_{p_3l_3}\delta_{p_4q_4}\delta_{q_1k_1}\delta_{q_2k_2}\delta_{q_3k_3}\delta_{k_4l_4}\delta_{l_1m_1}\delta_{m_2n_2}\delta_{m_3n_3}\delta_{m_4n_4}$,
		\item $\prod_k \delta(c^{4;1}_k(\vec{p}_l)) = \delta_{p_1l_1}\delta_{p_2q_2}\delta_{p_3q_3}\delta_{p_4q_4}\delta_{q_1k_1}\delta_{k_2l_2}\delta_{k_3l_3}\delta_{k_4l_4}$,
		\item $\prod_k \delta(c^{4;2}_k(\vec{p}_l)) = \delta_{p_1q_1}\delta_{p_2q_2}\delta_{p_3q_3}\delta_{p_4q_4}\delta_{k_1l_1}\delta_{k_2l_2}\delta_{k_3l_3}\delta_{k_4l_4}$,
	\end{itemize}
	where $\vec{p},\vec{k},\vec{m}$ are incoming and $\vec{q},\vec{l},\vec{n}$ outgoing momenta. Notice that the vertex functions are not symmetric with respect to all permutations of the edges, so in addition to specifying the type of vertices, one must also specify the orientations. (This is naturally taken care of in the stranded representation.)
	\item Sum over $\vec{p}_e\in\mb{Z}^4$ associated to the internal edges $e \in \Gamma^{[1]}_{int}$ of the graph.
\end{enumerate}
The quantity that results from applying these rules to a Feynman graph is the Feynman amplitude of that graph.

Since the momenta are conserved by a simple identification of variables in the propagator and the vertices, as mentioned above, the momentum vector components may be associated to the individual strands $s\in \Gamma^{s}$ of a Feynman graph in the stranded representation. We emphasize that this combinatorial identification of field variables is, in fact, what allows for the stranded representation in the first place. As a result, the Feynman amplitude of a stranded graph consists of a product over Kronecker deltas identifying the external momentum vector components through the strands, multiplied by the product over internal propagators of the graph, which depend on the momenta associated to the strands that pass through them, summed over the momenta associated to internal (looped) strands. For example, for the Feynman graph in Fig.\ \ref{fig:BGRgraph1} the amplitude reads
\begin{eqnarray}
	\phi(\Gamma_{\mt{Fig.\ \ref{fig:BGRgraph1}}}) :=& \lambda_{6;1}\lambda_{6;2}\ \delta_{k_1 q_1} \delta_{k_2 l_2} \delta_{k_3 l_3} \delta_{k_4 l_4} \delta_{l_1 m_1} \delta_{m_2 n_2} \delta_{m_3 n_3} \delta_{m_4 n_4} \delta_{n_1 p_1} \delta_{p_2 q_2} \delta_{p_3 q_3} \delta_{p_4 q_4} \qquad\nonumber\\
	& \times \sum_{\substack{r_2,r_3,\\r_4\in\mb{Z}}} \left(\frac{1}{k_1^2 + r_2^2 + r_3^2 + r_4^2 + m^2}\right) \sum_{s_4\in\mb{Z}} \left(\frac{1}{l_1^2 + p_2^2 + p_3^2 + s_4^2 + m^2}\right) \left(\frac{1}{k_1^2 + p_2^2 + p_3^2 + s_4^2 + m^2}\right),
\end{eqnarray}
where $\vec{k},\vec{l},\vec{m},\vec{n},\vec{p},\vec{q}\in{\mathbb Z}^4$ 
are the momenta associated to the external edges
and $r_i$ ($i=2,3,4$) and $s_4$ are the momentum vector components associated to the internal strands. 
We observe that for large absolute values of the internal momenta $r_2,r_3,r_4,s_4$ the sums behave asymptotically as
\ba
	\phi(\Gamma_{\mt{Fig.\ \ref{fig:BGRgraph1}}}) \propto \sum_{\substack{r_2,r_3,\\r_4\in\mb{Z}}} \frac{1}{r_2^2 + r_3^2 + r_4^2}\ \sum_{s_4\in\mb{Z}} \frac{1}{s_4^4} \approx \int_{\mb{R}^3} \frac{\dd^3r}{r^2} \int_{\mb{R}} \frac{\dd s}{s^4} \propto \Lambda \times  \Lambda^{-3} = \Lambda^{-2}\,,
\ea
where $\Lambda$ is a cut-off on the momentum summation. Here, the first factor leading to a linear divergence in the limit $\Lambda\rightarrow\infty$ corresponds to the three strand loops in the upper left part of the graph and the second convergent factor corresponds to the single strand loop in the middle. Although the first factor leads to a subdivergence, the superficial divergence degree of the total graph $\Gamma_{\mt{Fig.\ \ref{fig:BGRgraph1}}}$ is $-2$, and therefore the graph is superficially convergent.

Similarly, the Feynman amplitude for the graph in Fig.\ \ref{fig:BGRnonstrandedgraph} reads
\begin{eqnarray}
	\phi(\Gamma_{\mt{Fig.\ \ref{fig:BGRnonstrandedgraph}}}) := & \lambda_{6;1}\lambda_{6;2}\ \delta_{k_1 l_1} \delta_{k_2 l_2} \delta_{k_3 l_3} \delta_{k_4 l_4} \delta_{m_1 n_1} \delta_{m_2 n_2} \delta_{m_3 n_3} \delta_{m_4 q_4} \delta_{n_4 p_4} \delta_{p_1 q_1} \delta_{p_2 q_2} \delta_{p_3 q_3} \nonumber\\
	 & \times \sum_{\substack{r_2,r_3,\\r_4\in\mb{Z}}} \left(\frac{1}{k_1^2 + r_2^2 + r_3^2 + r_4^2 + m^2}\right) \left(\frac{1}{l_1^2 + p_2^2 + p_3^2 + m_4^2 + m^2}\right) \left(\frac{1}{k_1^2 + p_2^2 + p_3^2 + m_4^2 + m^2}\right),
\end{eqnarray}
with the same notation for the momenta as in Fig.\ \ref{fig:BGRgraph1}. The difference to the previous is that now there are only three internal strands corresponding to $r_i$, $i=2,3,4$. The superficial divergence degree of this graph is $-3$ despite the subdivergence, since the propagators contribute negatively to the power counting. In the following we state the power counting more explicitly.

The Feynman amplitude of any other BGR tensor graph may be computed in an analogous way. We leave as an exercise to the interested reader the calculation of amplitudes for the rest of the graphs of this section.

We may also draw the Feynman graphs as ordinary graphs, but then the combinatorics of the internal structure must be taken into account by additional labeling of vertices. We have illustrated the correspondence between the ordinary and the stranded representation of Feynman graphs of the BGR model in Fig.~\ref{fig:BGRnonstrandedgraph}.
\begin{figure}
\centering
\def\svgwidth{0.8\columnwidth}
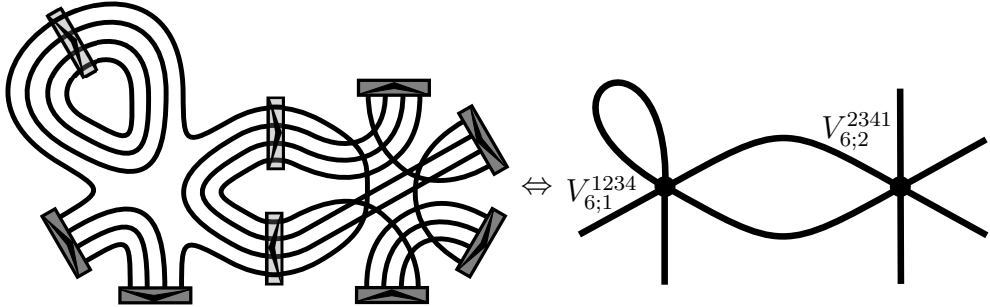
\caption{\label{fig:BGRnonstrandedgraph}An example of the correspondence between stranded and non-stranded representation of a Feynman graph of the BGR model.}
\end{figure}
The vertices are labeled, for example, as $V_{6;1}^{1234}$, where the lower index indicates the vertex type and the upper indices indicate the permutation of field variables. Since the vertex is not symmetric under permutations of edges, the edges meeting at a vertex must be distinguished. This can be accomplished, for example, by placing the vertex label between the first and the last edge, where the edges are ordered in a right-handed orientation around the vertex with respect to the normal pointing upwards from the surface, on which the graph is drawn.

Let us give the following definition:
\begin{definition}
	A \emph{face} of a stranded Feynman graph is a strand that forms a loop.
\end{definition}

Note that a face of a stranded graph is dual to the corresponding triangle of the tetrahedra of the four-dimensional triangulation.

The integrations over face momenta may give rise to divergencies, as we have already observed. In order for a strand to form a loop, the non-stranded representation must clearly have a loop as well. However, a loop in the non-stranded representation does not imply the existence of a face as can be seen, e.g., in Fig.~\ref{fig:BGRgraphexample3}.
\begin{figure}
\centering
\def\svgwidth{0.8\columnwidth}
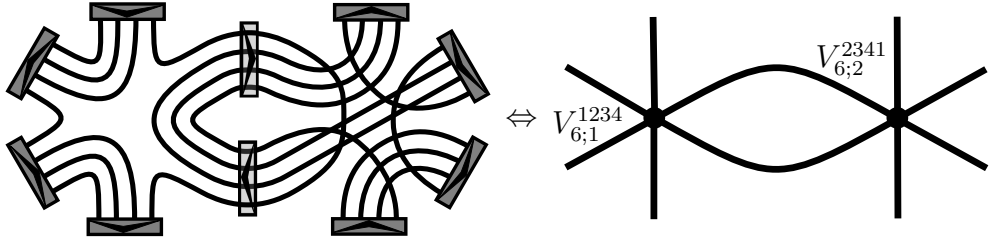
\caption{\label{fig:BGRgraphexample3}An example of a graph without a face but with a loop in the non-stranded representation.}
\end{figure}
Moreover, a face may extend around several loops of the non-stranded graph, as in Fig.~\ref{fig:BGRgraphexample4}.
\begin{figure}
\centering
\def\svgwidth{0.8\columnwidth}
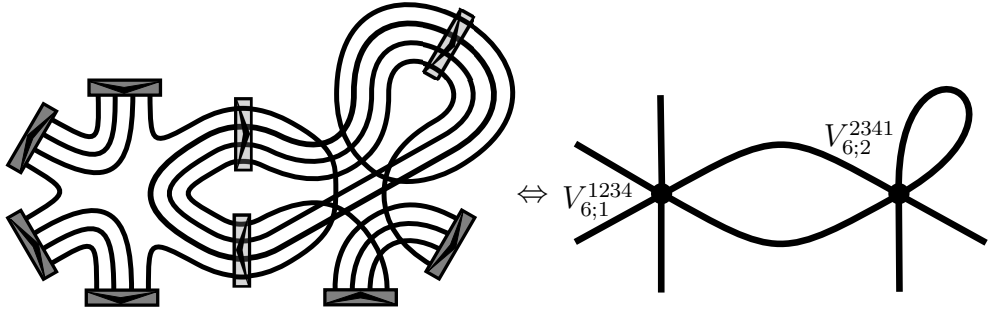
\caption{\label{fig:BGRgraphexample4}An example of a graph with faces, which extend around several loops of the non-stranded representation.}
\end{figure}
Therefore, the faces can be highly non-local objects on the graph. Fortunately, the power counting restricts the non-locality of divergent faces significantly, as we will see in the following. 
Notice also that faces cannot extend over bridges of the non-stranded Feynman graph, since due to the coloring a single strand can pass the same propagator only once. Therefore, 1PI Feynman graphs are still sufficient to resolve the divergence structure of Feynman amplitudes of the BGR model.

One has:
\begin{lemma}[\cite{BR}, Lemma 2]
	A superficial divergence degree for the BGR model is given by $\omega(\Gamma) = F_\Gamma - 2 P_\Gamma$, where $F_\Gamma$ and $P_\Gamma$ are the numbers of faces and propagators in $\Gamma$, respectively.
\end{lemma}
Thus, we see that not all faces lead to divergencies, but there is a balance between the number of propagators and the number of strands. It was further proved that the superficial divergence degree $\omega$ for connected graphs, where $V_{4;2}$ is considered disconnected, can be written in terms of the combinatorial and topological structure of the stranded Feynman graphs as follows.
\begin{definition}[\cite{BR}, Definition 1]
	Let $\Gamma$ be a Feynman graph of the BGR model. Then
	\begin{itemize}
		\item the \emph{colored extension} $\Gamma_{col}$ of $\Gamma$ is the corresponding unique stranded graph obtained by replacing the vertices of $\Gamma$ with the corresponding graphs of the simplicial colored model from Fig.~\ref{fig:BGRvertices6} and Fig.~\ref{fig:BGRvertex41}.
		\item a \emph{jacket} $J$ of $\Gamma_{col}$ is a ribbon graph obtained from $\Gamma_{col}$ by removing a subset of the strands. Specifically, a jacket $J$ is determined by a permutation $(0abcd)$ of colors (up to cyclic permutations), and the strands that are included in $J$ are the ones that connect the tetrahedra of colors $(0a),(ab),(bc),(cd),(d0)$, where $a,b,c,d \in \{R,G,B,Y\}$ are all different. Here, we have chosen 0 to represent the grey color, while the other colors are labeled by their initials.
		\item the \emph{pinched jacket} $\tilde{J}$ is the vacuum ribbon graph obtained from the jacket $J$ by closing all external ribbon edges of $J$. (See Fig.~\ref{fig:jacket} for an example.)
\begin{figure}
\centering
\def\svgwidth{0.95\columnwidth}
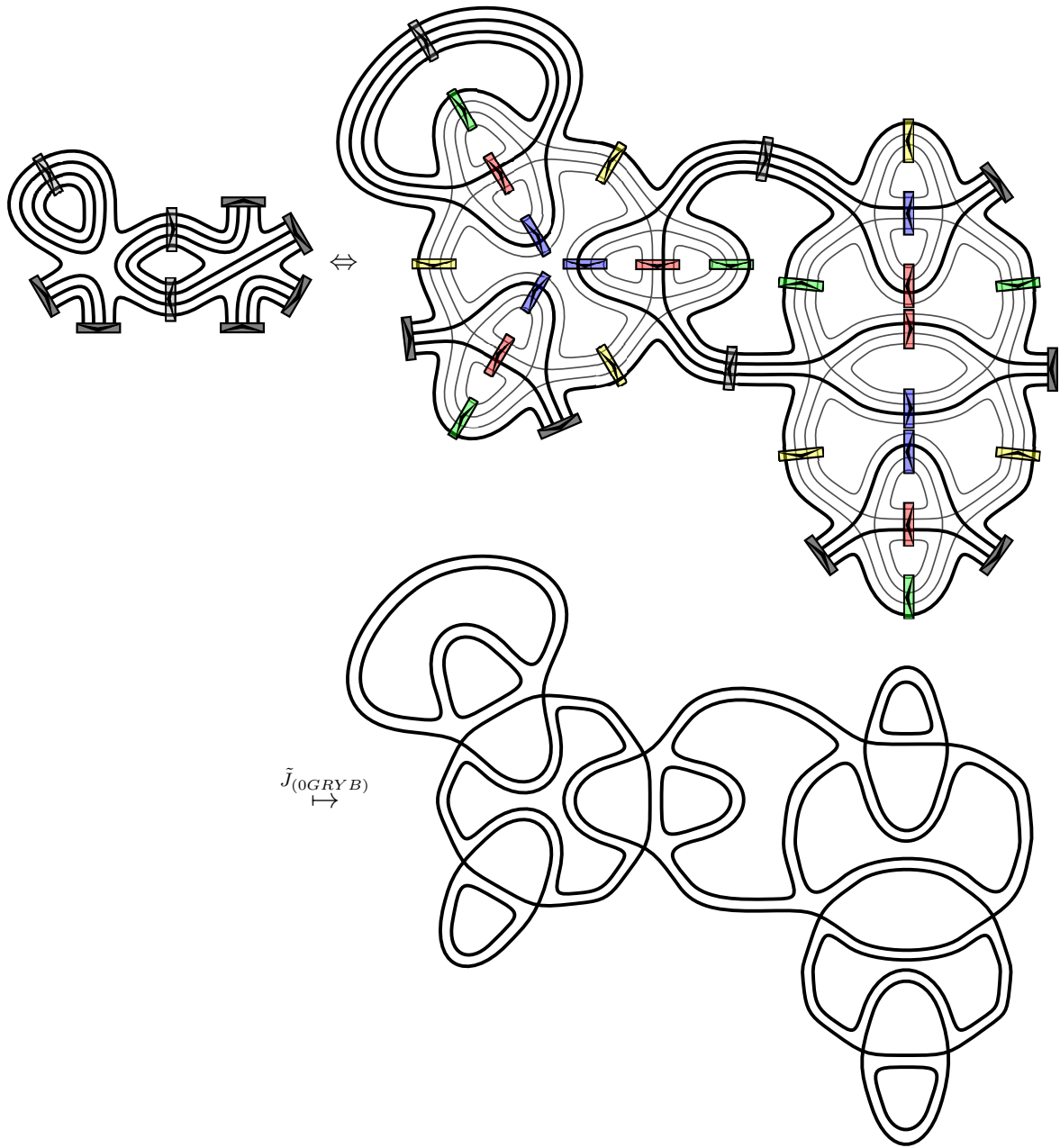
\caption{\label{fig:jacket}An example of a pinched jacket.}
\end{figure}
		\item the \emph{boundary graph} $\prt\Gamma$ of $\Gamma$ is the 3-dimensional tensor graph, which corresponds to the boundary of $\Gamma$. It is obtained by first removing all internal faces of $\Gamma$, closing the external edges by connecting their strands to a vertex per external edge, replacing the strands by bundles of three strands, and finally replacing the vertices of external edges by tensorial simplicial vertices. (See Fig.~\ref{fig:bound} for an example.)
\begin{figure}
\centering
\def\svgwidth{0.7\columnwidth}
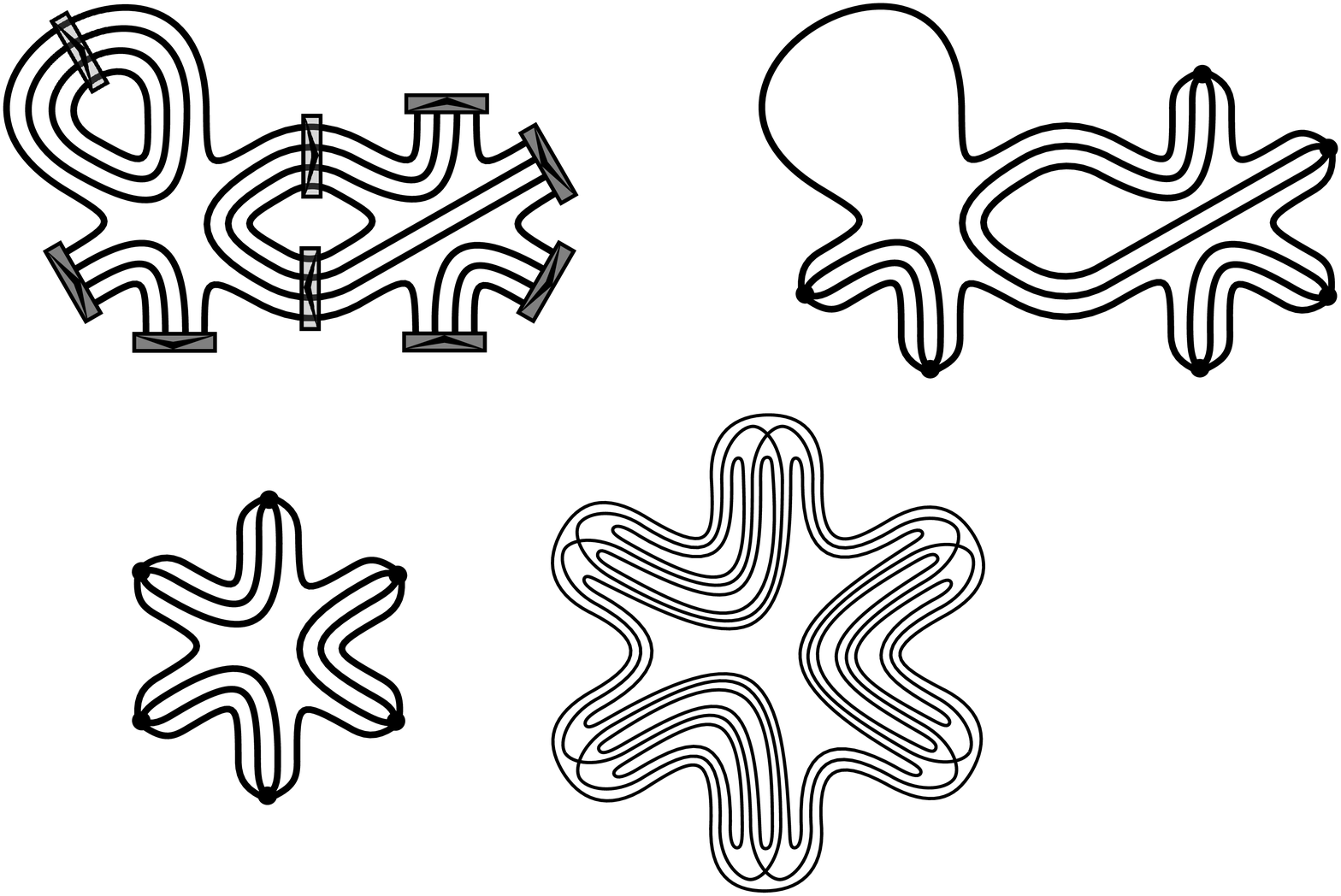
\caption{\label{fig:bound}An example of a boundary graph.}
\end{figure}
		\item the \emph{boundary jacket} $J_\prt$ is a jacket of the boundary graph $\prt\Gamma$.
	\end{itemize}
\end{definition}
\noindent
(See \cite{BR} and references therein for further details.) Ben Geloun and Rivasseau proved the following result:
\begin{theorem}
The superficial divergence degree $\omega_{BGR}$ of a connected Feynman graph $\Gamma$ can be written as
\ba
	\omega_{BGR}(\Gamma) = -\frac{1}{3}\left[ \sum_J g_{\tilde{J}} - \sum_{J_\prt} g_{J_\prt} \right] - (C_\prt -1) - V_4 -2V''_2 - \frac{1}{2}(N_{ext} - 6) \,,
\ea
where $g_{J}$ denotes the genus of the ribbon graph $J$, $C_\prt$ is the number of disconnected components of $\prt\Gamma$, $V_4$ is the number of vertices of type $V_{4;1}$, 
$V''_2/2$ is the number of vertices of type $V_{4;2}$, and $N_{ext}$ is the number of external edges of $\Gamma$.
\end{theorem}
Moreover, they classified all superficially divergent Feynman graphs in Table \ref{tab:supdivclass}, which is the departure point for our subsequent analysis.
\begin{table}[ht!]
\begin{center}
\begin{tabular}{c||c|c|c|c|c||c|c}
	class & $N_{ext}$ & $V_4$ & $\sum_{J_\prt} g_{J_\prt}$ & $C_{\prt} -1$ & $\sum_{\tilde{J}} g_{\tilde{J}}$ & $\omega_{BGR}$ & $\res$\\ \hline\hline
	$6_1$ & 6 & 0 & 0 & 0 & 0 & 0 & $V_{6;i}$ \\\hline
	$4_1$ & 4 & 0 & 0 & 0 & 0 & 1 & $V_{4;1}^*$ \\
	$4_2$ & 4 & 1 & 0 & 0 & 0 & 0 & $V_{4;1}$ \\
	$4_3$ & 4 & 0 & 0 & 1 & 0 & 0 & $V_{4;2}$ \\\hline
	$2_1$ & 2 & 0 & 0 & 0 & 0 & 2 & $V_{2;1}$ \\
	$2_2$ & 2 & 1 & 0 & 0 & 0 & 1 & $V_{2;2}$ \\
	$2_3$ & 2 & 2 & 0 & 0 & 0 & 0 & $V_{2;3}$ \\
	$2_4$ & 2 & 0 & 0 & 0 & 6 & 0 & $V_{2;3}$ \\\hline\hline
\end{tabular}
\end{center}
\caption{\label{tab:supdivclass}Classification and residues of superficially divergent Feynman graphs of the BGR model.}
\end{table}

Notice that all superficially divergent graphs have a combinatorial structure that corresponds to one of the vertices of the original model, which facilitates perturbative renormalizability. Ben Geloun and Rivasseau \cite{BR} showed that the BGR model is perturbatively renormalizable. 

\newpage
\section{Hopf algebraic description of the combinatorics of the renormalizability 
of the BGR model}
\label{sec:CKforBGR}
The formal definitions for local QFT models from Subsection \ref{subsec:HopfFeynman} apply also to the case of the BGR model with four important exceptions:

(i) The constraints on the external parameters of a Feynman graph $\Gamma$ are not a simple momentum conservation anymore, but a combinatorial identification of variables represented by the stranded structure of the Feynman diagrams. Let us give the following definition:
\begin{definition}
	A Feynman graph $\Gamma$ of a QFT model is equipped with a set of \emph{external constraints} $\{c^\Gamma_k(\vec{p}_i) = 0\}$ on the external momenta. We denote the space of the external momenta of $\Gamma$ by $P_\Gamma \equiv \{ (\vec{p}_i) \in \times_i P_\Gamma^i : c^\Gamma_k(\vec{p}_i) =0\}$, where $\vec{p}_i\in P_i$. The Feynman amplitudes are then generalized functions on $P_\Gamma$, as before.
\end{definition}
\begin{remark}
Notice that this generalizes the case of local QFT, where the external constraints impose the momentum conservation. In contrast, in tensor field theories the interactions are combinatorially non-local in the sense that the boundary variables are identified according to tensorial invariance. Thus, the notion of tensorial invariance in TFT substitutes for the notion of locality in local QFT. A gluing of tensor invariant graphs always yields another tensor invariant graph.
\end{remark}
Accordingly, we also need to generalize the notion of residue for tensorial Feynman graphs.
\begin{definition}
A residue $\res(\Gamma)$ of a tensorial Feynman graph $\Gamma\in\mc{G}_{1\mt{PI}}$ of the BGR model is the graph obtained, in the non-stranded representation, by shrinking all internal lines of $\Gamma$ to a point, and retaining the external constraints and structure of the graph. In the stranded representation this corresponds to the removal of all internal faces and internal propagators of $\Gamma$.
\end{definition}

(ii) There are vertices of valence four and six. By inserting a self-loop into a 6-valent vertex $V_{6;1}$, we obtain a self-loop Feynman graph with four external edges and divergence degree 1. This compels us to introduce a new 4-valent counter-term vertex $V_{4;1}^*$ that has the same structure as $V_{4;1}$ but a linear counter-term attached. This may be illustrated as in Fig.~\ref{fig:BGRgraph4pointexample}.

(iii) Similarly, arising from the contraction of superficially divergent subgraphs with two external edges, we introduce 2-valent counter-term vertices $V_{2;1}$, $V_{2;2}$ and $V_{2;3}$ with a quadratic, linear and logarithmic counter-term attached, respectively. This is very different, for example, from the $\phi^4$-model, for which we only have quadratically divergent 2-point functions. See Fig.~\ref{fig:BGRgraph2pointexample} for illustration. Of the 2-valent counter-term vertices, only the one with a quadratic divergence can be replaced with the original propagator, since only in this case the divergence degree of the graph is preserved.
\begin{figure}
\centering
\def\svgwidth{0.35\columnwidth}
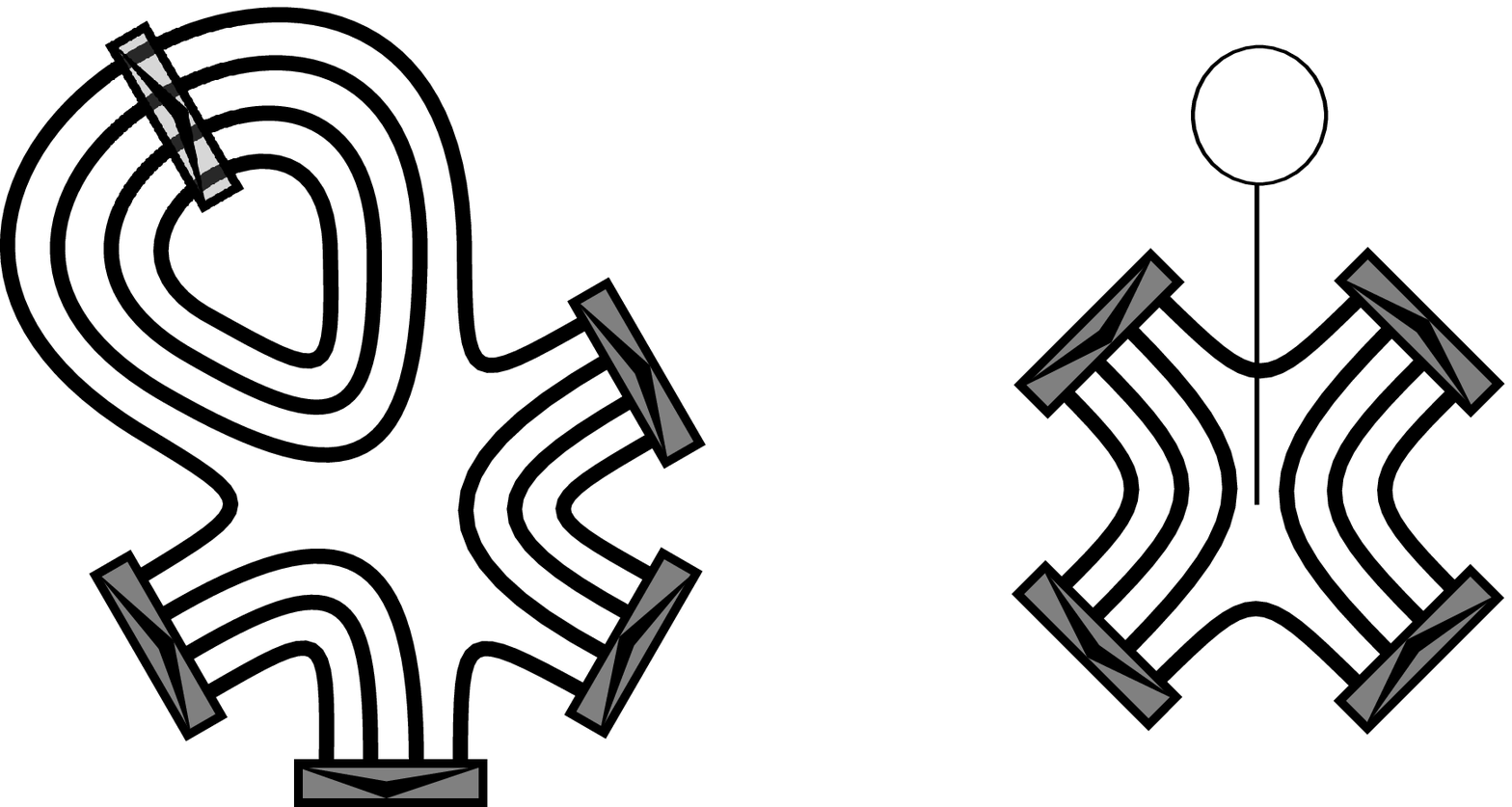
\caption{\label{fig:BGRgraph4pointexample}Illustration of the linear 4-valent counter-term vertex.}
\vspace{40pt}
\def\svgwidth{0.7\columnwidth}
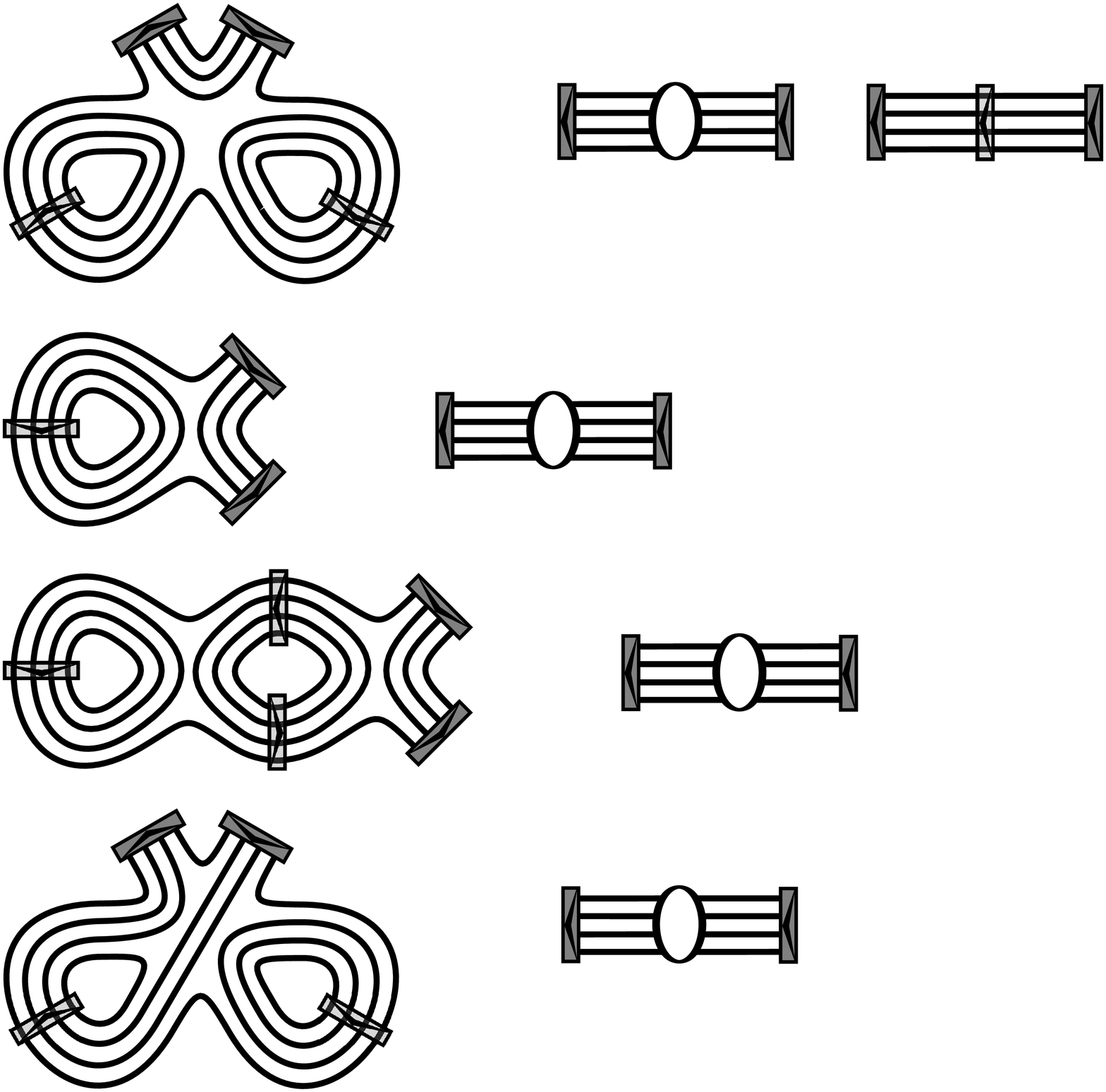
\caption{\label{fig:BGRgraph2pointexample}Some examples of the 2-valent counter-term as residues of 2-point functions from the classes $2_1$, $2_2$, $2_3$ and $2_4$, in the respective order. Only the quadratic counter-term vertex may be inserted into a propagator, since this conserves the divergence degree.}
\end{figure}

Let us make here a few more remarks on the renormalizability of the BGR model. 
When one does the Taylor development of the two-point function (see Lemmas $11$ and $12$ 
of \cite{BR} for the respective analytical details), the subleading divergence of the first graph in 
Fig. \ref{fig:BGRgraph2pointexample} is logarithmical and renormalizes the wave function.
Note that there is no linear subleading divergence. 
Analogously to the ``usual'' $\lambda \phi^4$ model (see Section $2$), there is no distinction between these two counterterms (quadratic, mass and logarithmic, wave function) at the level of the graph
drawing, but this is given by the external structure.
Moreover, there are no subleading divergencies for the rest of the graphs of Fig. \ref{fig:BGRgraph2pointexample}.

\begin{remark}
Note that a given two-point graph cannot have several counterterms of type $0$, $1$ and $2$ (see 
again Fig. \ref{fig:BGRgraph2pointexample}), since one can identify to which line of the Table 
\ref{tab:supdivclass}
 the respective tensor graph belongs (thanks to the numbers $V_4$, $V''_2$ and to the graph topology). If the graph belongs to the line $2_1$ then it has a counterterm of type $2$ (which means a quadratic counterterm of mass type and a logarithmic counterterm of wave function-type, see above), if it belongs to the line $2_2$ then it has a counterterm of type $1$ and so on.
\end{remark}

(iv) The superficial divergence degree does not only depend on the external data of a Feynman graph, such as the number of external edges, but also some internal combinatorial and topological information, in particular, the number of 4-valent vertices and the sum over the genera of pinched jackets. Of these, the first one is easy enough to understand, but to have a control over the second one requires a bit more work. To this aim, we prove the following:

\begin{lemma}
Let $\Gamma\in\mc{G}$, $\gamma\in\mc{G}_{sd}$ be Feynman graphs of the BGR model such that $\gamma\subseteq\Gamma$. We have for the sum over the genera of pinched jackets the relation
\ba
	\sum_J g_{\tilde{J}}(\Gamma) = \sum_J g_{\tilde{J}}(\gamma) + \sum_J g_{\tilde{J}}(\Gamma/\gamma) \,.
\ea
\end{lemma}
\begin{proof}
First note that the genus for a single pinched jacket may be expressed via the Euler characteristic formula as $g_{\tilde{J}} = -\frac{1}{2}(F_{\tilde{J}}-L_{\tilde{J}}+V_{\tilde{J}} -2C_{\tilde{J}})$, where $F_{\tilde{J}}$, $L_{\tilde{J}}$, $V_{\tilde{J}}$ and $C_{\tilde{J}}$ are the numbers of faces, lines, vertices and connected components of $\tilde{J}$. (For a disconnected jacket we define the genus to be given by the sum over the genera of its connected components.) One may further write $F_{\tilde{J}} = F_{\tilde{J},ext} + F_{\tilde{J},int}$, where $F_{\tilde{J},ext}$ is the number of faces formed by pinching the external edges of $J$, and $F_{\tilde{J},int}$ is the number of (internal) faces of $J$. Now, one may easily verify that each of $F_{\tilde{J},int}$, $L_{\tilde{J}}$ and $V_{\tilde{J}}$ separately satisfies the relation
\ba
	Q(\Gamma) - Q(\gamma) = Q(\Gamma/\gamma) - Q(\res(\gamma))
\ea
with respect to the contraction of subgraphs $\gamma\in\mc{G}_{sd}$ as defined for Feynman graphs of the BGR model. (Notice, in particular, that in general $F_{\tilde{J},int}(\res(\gamma))\neq 0$, $L_{\tilde{J}}(\res(\gamma))\neq 0$ and $V_{\tilde{J}}(\res(\gamma))\neq 0$, since the vertices of the BGR model correspond to extended subgraphs in the colored extension of the Feynman graphs, from which the jackets are derived.) For example, it is clear that both sides of
\ba
	Q(\Gamma) - Q(\Gamma/\gamma) = Q(\gamma) - Q(\res(\gamma))
\ea
represent the number of internal faces, lines or vertices, which are removed from $\tilde{J}$ by contracting $\gamma\subseteq\Gamma$, since by definition the subgraph contraction operation replaces $\gamma$ by $\res(\gamma)$ inside $\Gamma$, which does not affect the numbers of internal faces, lines or vertices that do not belong in $\gamma$. (Notice that some of the internal faces of $\Gamma$ not internal to $\gamma$ may still pass through $\gamma$, but are not affected by the contraction of $\gamma$.) Accordingly, we have
\ba
	Q(\Gamma) - Q(\res(\gamma)) = Q(\gamma) - Q(\res(\gamma)) + Q(\Gamma/\gamma) - Q(\res(\gamma)) \,,
\ea
so $G(\Gamma):= Q(\Gamma) - Q(\res(\gamma))$ gives a grading of jackets with respect to the contraction of $\gamma$ in the sense that $G(\Gamma)=G(\gamma)+G(\Gamma/\gamma)$. Moreover, if we have a set of gradings $G_i$, then it is easy to verify that $\sum_{i} (-1)^{n_i} G_i$ is also a grading for any choice of $n_i\in\{0,1\}$.

Now, note that we may write
\ba
	&\qquad g_{\tilde{J}}(\Gamma) + \frac{1}{2} \big(F_{\tilde{J},ext}(\Gamma) + F_{\tilde{J}}(\res(\gamma)) - L_{\tilde{J}}(\res(\gamma)) + V_{\tilde{J}}(\res(\gamma)) - 2C_{\tilde{J}}(\res(\gamma))\big) \nn
	&= -\frac{1}{2}\big((F_{\tilde{J},int}(\Gamma) - F_{\tilde{J},int}(\res(\gamma))) -(L_{\tilde{J}}(\Gamma) - L_{\tilde{J}}(\res(\gamma))) + (V_{\tilde{J}}(\Gamma) - V_{\tilde{J}}(\res(\gamma)))\big)  \,.
\ea
Thus, we find that $g_{\tilde{J}}(\Gamma) + \frac{1}{2}F_{\tilde{J},ext}(\Gamma) + X_{\tilde{J}}(\res(\gamma))$ is a grading with respect to the contraction of $\gamma$, where we denoted
\ba
	X_{\tilde{J}}(\res(\gamma)) := \frac{1}{2} \big(F_{\tilde{J},int}(\res(\gamma)) - L_{\tilde{J}}(\res(\gamma)) + V_{\tilde{J}}(\res(\gamma)) - 2C_{\tilde{J}}(\res(\gamma))\big) \,.
\ea
Accordingly, we find that, as a sum of gradings,
\ba
	\sum_{J} g_{\tilde{J}}(\Gamma) + \frac{1}{2}\sum_J F_{\tilde{J},ext}(\Gamma) + \sum_J X_{\tilde{J}}(\res(\gamma))
\ea
is also a grading with respect to the contraction of $\gamma$. Let us denote $X(\res(\gamma)) := \sum_{J} X_{\tilde{J}}(\res(\gamma))$. Moreover, $\sum_J F_{\tilde{J},ext}(\Gamma) = 2F_{\prt\Gamma}$, the number of faces of the boundary $\prt\Gamma$ of $\Gamma$. Then, due to the grading property, one has:
\ba
	\sum_J g_{\tilde{J}}(\Gamma) + F_{\prt\Gamma} = \sum_J g_{\tilde{J}}(\gamma) + F_{\prt\gamma} + \sum_J g_{\tilde{J}}(\Gamma/\gamma) + F_{\prt(\Gamma/\gamma)} + X(\res(\gamma)) \,.
\ea
We may further simplify by noting that $F_{\prt\Gamma}=F_{\prt(\Gamma/\gamma)}$, since a subgraph contraction does not affect the boundary, so we may write
\ba
	\sum_J g_{\tilde{J}}(\Gamma) &= \sum_J g_{\tilde{J}}(\gamma) + \sum_J g_{\tilde{J}}(\Gamma/\gamma) + F_{\prt\gamma}(\res(\gamma)) + X(\res(\gamma))\,.
\ea
But now
\ba
 F_{\prt\gamma}(\res(\gamma)) + X(\res(\gamma)) &= \frac{1}{2} \sum_{J} \big(F_{\tilde{J}}(\res(\gamma)) - L_{\tilde{J}}(\res(\gamma)) + V_{\tilde{J}}(\res(\gamma)) - 2C_{\tilde{J}}(\res(\gamma))\big) \nn
 &= -\sum_{J} g_{\tilde{J}}(\res(\gamma)) \,,
\ea
so we finally obtain
\ba
	\sum_J g_{\tilde{J}}(\Gamma) = \sum_J g_{\tilde{J}}(\gamma) + \sum_J g_{\tilde{J}}(\Gamma/\gamma) - \sum_J g_{\tilde{J}}(\res(\gamma)) \,.
\ea

One may then explicitly check that all the genera of pinched jackets vanish for different choices of $\res(\gamma)$ for $\gamma\in\mc{G}_{sd}$ using the Euler character formula for $g_{\tilde{J}}$. For example, for $\res(\gamma) = V_{6;1}$, one has $V_{\tilde{J}}(V_{6;1})=6$, $L_{\tilde{J}}(V_{6;1})=12$ and $C_{\tilde{J}}(V_{6;1})=1$ for all $J$, so $g_{\tilde{J}}=-\frac{1}{2}(F_{\tilde{J}}(V_{6;1})-8)$. One must then count the number of faces for each jacket, which gives $F_{\tilde{J}}(V_{6;1}) = 8$ for all $J$, so $g_{\tilde{J}}(V_{6;1}) = 0$ for all $J$. Similar calculations may be done for the other possible residues of superficially divergent graphs, which results in the above statement.
\end{proof}


With the above
generalizations of the usual Connes-Kreimer framework, and the power-counting for the model given in Table \ref{tab:supdivclass}, we are then ready to define the Hopf algebra of BGR model Feynman graphs in the same fashion as above for local QFT. In particular, we define the coproduct for the 1PI Feynman graphs $\Gamma\in\mc{G}_{\mt{1PI}}$ as
\ba\label{eq:coprod2}
	\Delta(\Gamma) = \Gamma \otimes \1 + \1 \otimes \Gamma + \sum_{\substack{\gamma\in\cup\mc{G}_{sd}^\omega\\\gamma\subsetneq\Gamma}} \gamma \otimes \Gamma/\gamma \,.
\ea
where the subgraph contraction is defined as in Definition \ref{def:subcontr}. Recall that the summation runs over the set $\cup\mc{G}_{sd}^\omega$ of all disconnected unions of superficially divergent 1PI subgraphs. (See Fig.~\ref{fig:BGRgraphcoprodexample} for an example.)
\begin{figure}
\centering
\def\svgwidth{0.9\columnwidth}
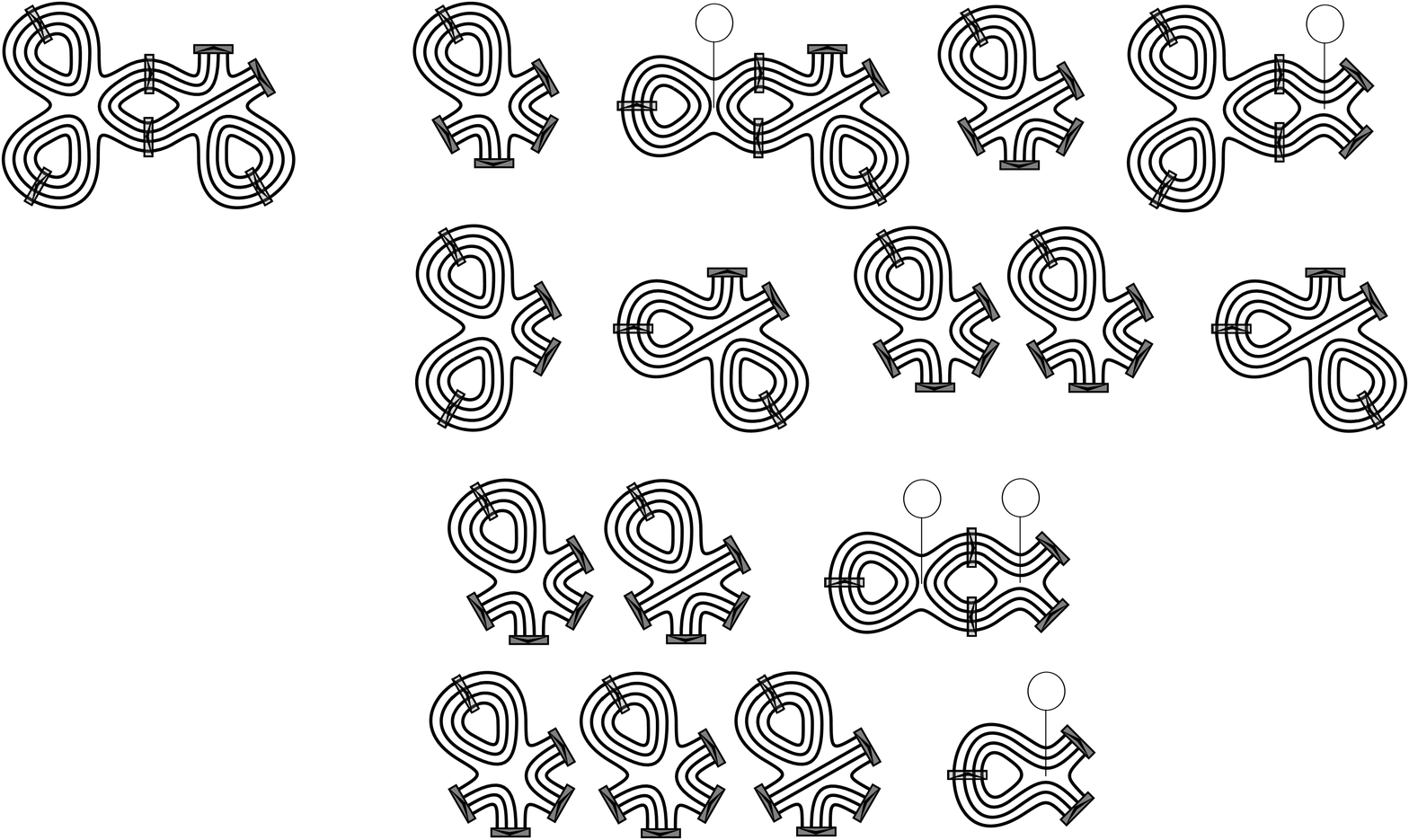
\caption{\label{fig:BGRgraphcoprodexample}An example of the coproduct $\Delta(\Gamma) \equiv \Gamma \otimes \1 + \1 \otimes \Gamma + \Delta'(\Gamma)$.}
\end{figure}

One has:
\begin{lemma}
	The set of superficially divergent Feynman diagrams of the BGR model is closed under the operations of subgraph contraction and insertion.
\end{lemma}
\begin{proof}
	Let us first consider the subgraph contraction on a class-by-class basis. Notice that a subgraph contraction does not affect the boundary of a graph, and therefore we need not consider the boundary properties. Accordingly, the contractions of subgraphs of a certain class give the following:
\begin{itemize}
	\item[$6_1$:] Contracting a subgraph in $6_1,4_1,4_3,2_1$ gives a graph in $6_1$. Other classes are not possible, since $6_1$ does not contain 4-valent vertices, and the genuses of its pinched jackets are zero.
	\item[$4_1$:] Contracting a subgraph in $6_1,4_1,4_3,2_1$ gives a graph in $4_1$. Other classes are not possible for the same reasons as above.
	\item[$4_2$:] Contracting a subgraph in $6_1,4_1,4_2,4_3,2_1$ gives a graph in $4_2$. Contracting a subgraph in $2_2$ gives a graph in $4_1$. Other subgraph classes are not possible.
	\item[$4_3$:] Contracting a subgraph in $6_1,4_1,4_3,2_1$ gives a graph in $4_3$. Other subgraph classes are not possible.
	\item[$2_1$:] Contracting a subgraph in $6_1,4_1,4_3,2_1$ gives a graph in $2_1$. Other subgraph classes are not possible.
	\item[$2_2$:] Contracting a subgraph in $6_1,4_1,4_2,4_3,2_1$ gives a graph in $2_2$. Contracting a subgraph in $2_2$ gives a graph in $2_1$. Other subgraph classes are not possible.
	\item[$2_3$:] Contracting a subgraph in $6_1,4_1,4_2,4_3,2_1$ gives a graph in $2_3$. Contracting a subgraph in $2_2$ gives a graph in $2_2$. Contracting a subgraph in $2_3$ gives a graph in $2_1$. Other subgraph classes are not possible.
	\item[$2_4$:] Contracting a subgraph in $6_1,4_1,4_3,2_1$ gives a graph in $2_4$. Contracting a subgraph in $2_4$ gives a graph in $2_1$. Other subgraph classes are not possible.
\end{itemize}
We may condense the above analysis into Table \ref{tab3}.
\begin{table}
\begin{center}
\begin{tabular}{c||c|ccc|cccc}
	row/col & $6_1$ & $4_1$ & $4_2$ & $4_3$ & $2_1$ & $2_2$ & $2_3$ & $2_4$ \\\hline\hline
	$6_1$ & $6_1$ & $6_1$ & 0 & $6_1$ & $6_1$ & 0 & 0 & 0 \\\hline
	$4_1$ & $4_1$ & $4_1$ & 0 & $4_1$ & $4_1$ & 0 & 0 & 0 \\
	$4_2$ & $4_2$ & $4_2$ & $4_2$ & $4_2$ & $4_2$ & $4_1$ & 0 & 0 \\
	$4_3$ & $4_3$ & $4_3$ & 0 & $4_3$ & $4_3$ & 0 & 0 & 0 \\\hline
	$2_1$ & $2_1$ & $2_1$ & 0 & $2_1$ & $2_1$ & 0 & 0 & 0 \\
	$2_2$ & $2_2$ & $2_2$ & $2_2$ & $2_2$ & $2_2$ & $2_1$ & 0 & 0 \\
	$2_3$ & $2_3$ & $2_3$ & $2_3$ & $2_3$ & $2_3$ & $2_2$ & $2_1$ & 0 \\
	$2_4$ & $2_4$ & $2_4$ & 0 & $2_4$ & $2_4$ & 0 & 0 & $2_1$ \\\hline\hline
\end{tabular}
\end{center}
\caption{\label{tab3}Various possibilities of contractions of superficially divergent Feynman graphs of the BGR model.}
\end{table}
We have marked a vanishing entry in the Table if a graph from the column class cannot be a subgraph of a graph from the row class.

\bigskip

Similarly, we have to consider the insertions of superficially divergent graphs into other superficially divergent graphs, as defined in Definition \ref{def:insert}. If we consider insertions into the original Feynman graphs of the model, which do not contain counter-term vertices, we must note the following: First, we may insert into the vertices only graphs, whose residue is equal to that vertex. Moreover, into bare propagators we may only insert superficially divergent graphs with two external edges from the class $2_1$, since these have divergence degree 2, which is needed in order to conserve the divergence degree of the graph under insertions. The other superficially divergent graphs with $N_{ext}=2$ can only be inserted into the corresponding counter-term vertices, which are not a part of the original bare model. With these restrictions in mind, we provide the following analysis of the insertions into the  graphs of the model:
\begin{itemize}
	\item[$6_1$:] Inserting a graph from the class $6_1,4_3,2_1$ yields a graph in the class $6_1$. Insertion of a graph in $4_2$ is not possible, since $6_1$ does not contain $V_{4;1}$ vertices.
	\item[$4_1$:] Inserting a graph from the class $6_1,4_3,2_1$ yields a graph in the class $4_1$. Insertion of a graph in $4_2$ is not possible, since $4_1$ does not contain $V_{4;1}$ vertices.
	\item[$4_2$:] Inserting a graph from the class $6_1,4_2,4_3,2_1$ yields a graph in the class $4_2$.
	\item[$4_3$:] Inserting a graph from the class $6_1,4_3,2_1$ yields a graph in the class $4_3$. Insertion of a graph in $4_2$ is not possible, since $4_3$ does not contain $V_{4;1}$ vertices.
	\item[$2_1$:] Inserting a graph from the class $6_1,4_3,2_1$ yields a graph in the class $2_1$. Insertion of a graph in $4_2$ is not possible, since $2_1$ does not contain $V_{4;1}$ vertices.
	\item[$2_2$:] Inserting a graph from the class $6_1,4_2,4_3,2_1$ yields a graph in the class $2_2$.
	\item[$2_3$:] Inserting a graph from the class $6_1,4_2,4_3,2_1$ yields a graph in the class $2_3$.
	\item[$2_4$:] Inserting a graph from the class $6_1,4_3,2_1$ yields a graph in the class $2_4$. Insertion of a graph in $4_2$ is not possible, since $2_4$ does not contain $V_{4;1}$ vertices.
\end{itemize}
Accordingly, all the allowed insertions lead to superficially divergent graphs, and the claim follows.
\end{proof}

Given the above Lemma, we may now follow along the lines of the coassociativity proof for the coproduct in \cite{TV}.
\begin{lemma}
	The coproduct defined in Equation (\ref{eq:coprod2}) is coassociative, if the set of superficially divergent Feynman diagrams is closed under the operations of subgraph contraction and insertion.
\end{lemma}
\begin{proof}
	First, by a simple calculation we find that
\ba
	(\Delta \otimes \id_\mc{H})\circ \Delta - (\id_\mc{H} \otimes \Delta)\circ \Delta = (\Delta' \otimes \id_\mc{H})\circ \Delta' - (\id_\mc{H} \otimes \Delta')\circ \Delta'
\ea
for $\Delta(\Gamma) \equiv \Gamma \otimes \1 + \1 \otimes \Gamma + \Delta'(\Gamma)$, and therefore
\ba
	(\Delta \otimes \id_\mc{H})\circ \Delta = (\id_\mc{H} \otimes \Delta)\circ \Delta \quad \Leftrightarrow\quad (\Delta' \otimes \id_\mc{H})\circ \Delta' = (\id_\mc{H} \otimes \Delta')\circ \Delta' \,.
\ea
Moreover, it is enough to show the coassociativity property for 1PI Feynman graphs, which are the generators of $\mc{H}$, since $\Delta'$ is an algebra homomorphism, i.e., $\Delta(\Gamma\cup\Gamma')=\Delta(\Gamma)\cup\Delta(\Gamma')$, which is also easy to verify by a direct calculation. Now, on the one hand we have
	\ba
		(\Delta' \otimes \id_\mc{H})\circ \Delta' \Gamma = \sum_{\substack{\gamma\in\cup\mc{G}_{sd}^\omega\\\gamma\subsetneq\Gamma}} \Delta'\gamma \otimes \Gamma/\gamma = \sum_{\substack{\gamma\in\cup\mc{G}_{sd}^\omega\\\gamma\subsetneq\Gamma}} \sum_{\substack{\gamma'\in\cup\mc{G}_{sd}^\omega\\\gamma'\subsetneq\gamma}} \gamma'\otimes \gamma/\gamma' \otimes \Gamma/\gamma \,.
	\ea
	On the other hand we have
	\ba
		(\id_\mc{H} \otimes \Delta')\circ \Delta' \Gamma = \sum_{\substack{\gamma'\in\cup\mc{G}_{sd}^\omega\\\gamma'\subsetneq\Gamma}} \gamma' \otimes \Delta'(\Gamma/\gamma') = \sum_{\substack{\gamma'\in\cup\mc{G}_{sd}^\omega\\\gamma'\subsetneq\Gamma}} \sum_{\substack{\gamma''\in\cup\mc{G}_{sd}^\omega\\\gamma''\subsetneq\Gamma/\gamma'}} \gamma' \otimes \gamma'' \otimes (\Gamma/\gamma')/\gamma'' \,.
	\ea
	For the coassociativity property, we must show that these are equal for all algebra generators $\Gamma\in\mc{G}_{\mt{1PI}}$. Notice first that in the first sum we may change the order of the two summations as
	\ba
		(\Delta' \otimes \id_\mc{H})\circ \Delta' \Gamma = \sum_{\substack{\gamma'\in\cup\mc{G}_{sd}^\omega\\\gamma'\subsetneq\Gamma}} \sum_{\substack{\gamma\in\cup\mc{G}_{sd}^\omega\\\gamma'\subsetneq\gamma\subseteq\Gamma}} \gamma'\otimes \gamma/\gamma' \otimes \Gamma/\gamma \,.
	\ea
	It is therefore enough to show that for a fixed $\gamma'$, we have
	\ba
		\sum_{\substack{\gamma\in\cup\mc{G}_{sd}^\omega\\\gamma'\subsetneq\gamma\subsetneq\Gamma}} \gamma/\gamma' \otimes \Gamma/\gamma = \sum_{\substack{\gamma''\cup\in\mc{G}_{sd}^\omega\\\gamma''\subsetneq\Gamma/\gamma'}} \gamma'' \otimes (\Gamma/\gamma')/\gamma'' \,.
	\ea
	These two expressions coincide, if we can set $\gamma''=\gamma/\gamma'$ in the summation for all $\gamma,\gamma',\gamma''\in\cup\mc{G}_{sd}^\omega$ such that $\gamma'\subsetneq\gamma\subsetneq\Gamma$ and $\gamma''\subsetneq\Gamma/\gamma'$, since then we have $(\Gamma/\gamma')/\gamma'' = (\Gamma/\gamma')/(\gamma/\gamma')=\Gamma/\gamma$. This amounts to two requirements:
	\begin{enumerate}
		\item Closure of $\mc{G}_{sd}^\omega$ under contraction: $\gamma/\gamma'\in\mc{G}_{sd}^\omega$ for all $\gamma,\gamma'\in\mc{G}_{sd}^\omega$ such that $\gamma'\subsetneq\gamma$, so that we can find $\gamma'' = \gamma/\gamma' \subsetneq\Gamma/\gamma'$ in $\cup\mc{G}_{sd}^\omega$.
		\item Closure of $\mc{G}_{sd}^\omega$ under insertion: For all $\gamma',\gamma''\in\mc{G}_{sd}^\omega$ we have $\gamma'\circ_v \gamma''\in\mc{G}_{sd}^\omega$ for any allowed insertion, so that we can find $\gamma = \gamma'\circ_v \gamma'' \subsetneq\Gamma$ in $\cup\mc{G}_{sd}^\omega$.
	\end{enumerate}
	As we have assumed these properties to hold, we conclude the proof.
\end{proof}

Our main theorem then follows:
\begin{theorem}
	Let $\mc{H}_{\mt{BGR}}$ be the unital associative algebra freely generated by the 1PI Feynman graphs of the Ben Geloun-Rivasseau TGFT model. Then $(\mc{H}_{\mt{BGR}},u,m,\epsilon,\Delta,S)$ is a Hopf algebra, where
	\begin{itemize}
		\item The unit $u:\mb{C}\rightarrow\mc{H}_{\mt{BGR}}$ is given by $u(1)=\1$, where $\1$ is the empty graph.
		\item The product $m:\mc{H}_{\mt{BGR}}\otimes\mc{H}_{\mt{BGR}}\rightarrow\mc{H}_{\mt{BGR}}$ is given by the the disjoint union of graphs.
		\item The counit $\epsilon:\mc{H}_{\mt{BGR}}\rightarrow\mb{C}$ is given by  $\epsilon(\1)=1$ and $\epsilon(\Gamma)=0$ for $\Gamma\neq\1$.
		\item The coproduct $\Delta:\mc{H}_{\mt{BGR}}\rightarrow\mc{H}_{\mt{BGR}}\otimes\mc{H}_{\mt{BGR}}$ is given by
			\ba
				\Delta(\Gamma) = \Gamma \otimes \1 + \1 \otimes \Gamma + \sum_{\substack{\gamma\in\cup\mc{G}_{sd}^\omega\\\gamma\subsetneq\Gamma}} \gamma \otimes \Gamma/\gamma \,.
			\ea
			which is extended to the whole of $\mc{H}_{\mt{BGR}}$ as an algebra homomorphism.
		\item The antipode $S:\mc{H}_{\mt{BGR}}\rightarrow\mc{H}_{\mt{BGR}}$ is given by the recursive formula
			\ba
				S(\Gamma) = -\Gamma - \sum_{\substack{\gamma\in\mc{G}_{sd}^\omega\\\gamma\subsetneq\Gamma}} S(\gamma) \Gamma/\gamma
			\ea
			with $S(\1)=\1$.
	\end{itemize}
\end{theorem}
\begin{proof}
	We have shown above that the coproduct is coassociative. All the other Hopf algebra properties follow easily. In particular, we note that $(\mc{H}_{\mt{BGR}},u,m,\epsilon,\Delta)$ is a graded connected bialgebra, graded by the number of internal edges. Therefore, the antipode follows from Formula (\ref{eq:Srecurs}).
\end{proof}

The 6-point function Feynman amplitude (without two- or four-point sub-divergencies)
of the Ben Geloun-Rivasseau model can then by formally identified to the general formula given in Theorem \ref{thm:ren}
\ba
	\phi_R(\Gamma) = S^\phi_R(\phi(\Gamma)) \star \phi(\Gamma) \,,
\ea
where
\ba
	S^\phi_R(\phi(\Gamma)) = - R[\phi(\Gamma)] - R\left[\sum_{\substack{\gamma\subsetneq\Gamma\\\gamma\in\mc{G}_{sd}^\omega}} S^\phi_R(\phi(\gamma)) \phi(\Gamma/\gamma) \right] \,,
\ea
is given through recursion. 
If one uses a dimensional regularization scheme, then the formula above can then be identified to the renormalized Feynman amplitude of any graph (recall however that in the 
original \cite{BR} paper, the position space multi-scale analysis Taylor expansion was used to prove renormalizability).

\section*{Acknowledgements}
The authors acknowledge Joseph Ben Geloun, R\u azvan Gur\u au and Vincent Rivasseau for
various discussions and suggestions.
The authors are partially supported by the "Combinatoire  alg\'ebrique" Univ. Paris 13, Sorbonne Paris Cit\'e BQR  grant and the "Cartes 3D" INS2I CNRS PEPS grant.
A. Tanasa is partially supported by 
the grant PN 09 37 01 02.
M. Raasakka gratefully acknowledges financial support from Emil Aaltonen Foundation.

\noindent
{\it\small $(a)$ LIPN, Institut Galil\'ee, CNRS UMR 7030, 
Universit\'e Paris 13, Sorbonne Paris Cit\'e,}\\{\it\small 99 av. Clement, 93430 Villetaneuse, France, EU}\\
{\it\small $(b)$ Horia Hulubei National Institute for Physics and Nuclear Engineering,
P.O.B. MG-6, 077125 Magurele, Romania, EU}

\end{document}